\documentclass[twoside]{article}

\usepackage[accepted]{aistats2026}

\usepackage{enumitem}
\usepackage{natbib}
\usepackage{makecell}

\usepackage{booktabs} 

\usepackage[toc,page,header]{appendix}
\usepackage{minitoc}

\usepackage{utils}

\begin{document}

\runningtitle{Convex Markov Games and Beyond}

\twocolumn[
\aistatstitle{Convex Markov Games and Beyond: New Proof of Existence, Characterization and Learning Algorithms for Nash Equilibria}

\aistatsauthor{Anas Barakat\footnotemark[1] \And Ioannis Panageas\footnotemark[2] \And  Antonios Varvitsiotis\footnotemark[1]\footnotemark[3]\footnotemark[4]}
\medskip
\aistatsaddress{SUTD \And UC Irvine \And SUTD, CQT NUS, Archimedes}
]

\begin{abstract}
Convex Markov Games (cMGs) were recently introduced as a broad class of multi-agent learning problems that generalize Markov games to settings where strategic agents optimize general utilities beyond additive rewards. 
While cMGs expand the modeling frontier, their theoretical foundations, particularly the structure of Nash equilibria (NE) and guarantees for learning algorithms, are not yet well understood. In this work, we address these gaps for an extension of cMGs, which we term General Utility Markov Games (GUMGs), capturing new applications requiring coupling between agents' occupancy measures. 
We prove that in GUMGs, Nash equilibria coincide with the fixed points of projected pseudo-gradient dynamics (i.e., first-order stationary points), enabled by a novel agent-wise gradient domination property. This insight also yields a simple proof of NE existence using Brouwer’s fixed-point theorem. We further show the existence of Markov \textit{perfect} equilibria. 
Building on this characterization, we establish a policy gradient theorem for GUMGs and design a model-free policy gradient algorithm. For potential GUMGs, we establish iteration complexity guarantees for computing approximate-NE under exact gradients and provide sample complexity bounds in both the generative model and on-policy settings. Our results extend beyond prior work restricted to zero-sum cMGs, providing the first theoretical analysis of common-interest cMGs. 
\end{abstract}

\footnotetext[1]{Singapore University of Technology and Design}
\footnotetext[2]{University of California Irvine, USA}
\footnotetext[3]{Centre for Quantum Technologies, National University of Singapore, Singapore}
\footnotetext[4]{Archimedes/Athena Research Center, Greece}

\section{INTRODUCTION}

Multi-agent reinforcement learning (MARL) is classically modeled by Markov games \citep{shapley53,fink64,littman94}, 
a direct generalization of Markov Decision Processes (MDPs) \citep{howard1960book,puterman14mdps}. While widely used, this additive-reward framework fails to capture key objectives such as risk sensitivity, imitation, exploration, or fairness. 

In the last few years, to go beyond additive rewards in MDPs, convex RL \citep{zhang-et-al20variational,zahavy-et-al21} has emerged as a general and unifying framework where the goal of an agent is to optimize a concave functional of its occupancy measure (i.e. the frequency of their state-action visitations under the selected policy). In particular, convex RL captures a variety of applications including pure exploration \citep{hazan-et-al19}, imitation learning \citep{ho-ermon16gail}, diverse skill discovery \citep{eysenbach-et-al19} and experiment design \citep{mutny-et-al23} to name a few. 

Inspired by single-agent convex RL, 
\citet{gemp-et-al25convex-mgs} recently  introduced convex Markov games (cMGs). 
Extending the scope of classical Markov games for MARL, cMGs provide a flexible unifying framework to model strategic interactions in MARL with general utilities beyond additive rewards. This framework unifies a wide array of MARL problems, from preference alignment \citep{wu-et-al24multistep} and task-agnostic exploration \citep{zamboni-et-al25} to risk-sensitive sequential games \citep{bhatt-sobel25} and fairness-aware MARL  \citep{hughes-et-al18}.

While cMGs extend the modeling frontier, their theoretical foundations remain poorly understood. 
\citet{gemp-et-al25convex-mgs} proved existence of Nash equilibria (NE) and proposed approximation methods based on regularization annealing, but required full model knowledge and left equilibrium structure uncharacterized.
More recently, \citet{kalogiannis-et-al25zs-mgs} analyzed the zero-sum subclass of cMGs, but efficient algorithms for general or common-interest cMGs remain largely open.  

This raises three fundamental questions: (i) what is the structure of NE in general cMGs beyond existence, (ii) can we design efficient policy gradient algorithms for such games, and (iii) can we learn NE sample-efficiently in the common-interest setting, where the agents' utilities have a potential structure?

\begin{figure}[t]
\centering
\begin{tikzpicture}[scale=0.7,every node/.style={font=\bfseries}]
\filldraw[fill=white, draw=black, thick] (0,0) ellipse (4.1 and 3);
\filldraw[fill=gray!15, draw=black, thick] (0,0) ellipse (2.8 and 2.1);
\filldraw[fill=gray!40, draw=black, thick] (0,0) circle (1.3);
\filldraw[fill=gray!60, draw=black, thick] (0,-0.3) circle (0.7);
\node[font=\bfseries]  at (0,2.6) {GUMG};
\node[font=\bfseries] at (0,1.8) {cMG};
\node[font=\bfseries] at (0,1) {MG};
\node[font=\bfseries \small] at (0,-0.3) {MDP};
\end{tikzpicture}
\caption{From single-agent RL to General Utility Markov Games (GUMG). MDP: Markov Decision Processes (single-agent RL), MG: Markov Games (MARL), cMG: convex Markov Games.}
\label{fig:gumgs}
\end{figure}
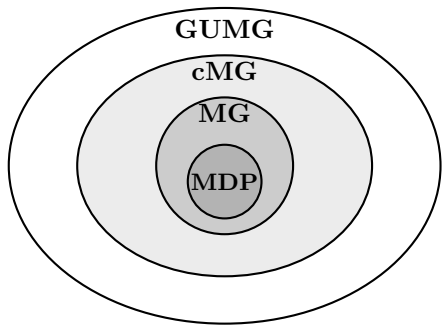

\subsection{Our Contributions}

In this work, we address these questions for an extension of cMGs. Our key contributions are: 
\begin{itemize}
\item \textbf{Framework.} We introduce the new class of General Utility Markov Games (GUMGs) as a natural extension of cMGs. 
This class strictly generalizes MDPs, Markov games, and cMGs (see Fig.~\ref{fig:gumgs}). 

\item \textbf{Existence and characterization of NE.} 
We prove that NE in GUMGs coincide with fixed points of projected pseudo-gradient dynamics. This yields a conceptually simple proof of equilibrium existence via a novel agent-wise gradient domination property of utilities in GUMGs.

\item \textbf{Policy gradient theorem and algorithm.} Motivated by this characterization, we establish a multi-agent policy gradient theorem for GUMGs and leverage it to design a principled model-free policy gradient (PG) algorithm  
(Algorithm~\ref{algo}). 

\item \textbf{Learning NE in potential GUMGs.} In potential GUMGs where utilities have a potential structure, our PG algorithm finds an $\varepsilon$-approximate NE in $\mathcal{O}(\varepsilon^{-2})$ iterations with exact gradients. When the transition dynamics are unknown, the PG algorithm learns an $\varepsilon$-approximate NE with total sample complexity of $\mathcal{O}(\varepsilon^{-4})$ under a generative model setting and $\mathcal{O}(\varepsilon^{-5})$ when only having access to on-policy sampled trajectories. To our knowledge, these are the first efficient NE learning guarantees for common-interest cMGs, beyond zero-sum cMGs. 
\end{itemize} 

\subsection{Technical Overview}

\noindent\textbf{Obstacles beyond Markov games.} 
Two obstacles prevent the direct extension of classical MG analysis to cMGs and our GUMGs. 
First, best response sets are non-convex in the policy space. Second, value functions and Bellman contraction arguments are not immediately available since the utilities are non-additive, unlike in Markov games. This hinders the application of Kakutani's fixed point theorem as used for the existence of NE in Markov games \citep{fink64}. 
Prior work on cMGs \citep{gemp-et-al25convex-mgs} established existence via Debreu’s theorem (using contractible sets) \citep{debreu52social-existence-eq-thm} to bypass the lack of convexity of best response sets, but provided no structural characterization of equilibria and assumed full model knowledge for (biased) approximate computation.

\noindent\textbf{Technical tool: per-agent gradient domination.} We prove a per-agent gradient domination property which holds for our general class of $N$-agents GUMGs (hence in cMGs). This per-player property implies that any (approximate) first-order stationary policy is also an (approximate) best response for that player. This generalizes gradient domination arguments used in special cases (MDPs/MGs; \citealt{agarwal-et-al21jmlr,daskalakis-foster-golowich20,leonardos-et-al22iclr,zhang-et-al24tac}) to non-additive, multi-agent utilities.

\noindent\textbf{Existence and first-order characterization via Brouwer.} 
Per-agent gradient domination in GUMGs implies the \textit{equivalence} between (a) NE, (b) first-order stationary points and (c) \textit{fixed-points of the projected pseudo-gradient dynamics}. Applying \textit{Brouwer's} fixed point theorem (rather than Kakutani or Debreu) to the projected pseudo-gradient map which is continuous on the compact policy space yields both existence and a structural characterization of NE in GUMGs (for any full-support initial state distribution), addressing the main theoretical gap left open by \citet{gemp-et-al25convex-mgs}. 
Therefore our gradient domination result allows us to circumvent the first obstacle above to (i) show existence of NE in GUMGs for any full support initial state distribution and (ii) characterize them.

\noindent\textbf{From NE to MPE.} 
We further prove existence of Markov-perfect equilibria (MPE) beyond the mere existence of \textit{non-perfect} NE (which depend on the initial state distribution). First, we show that the NE above does not depend on the choice of the  full-support initial distribution. Then, for initial distributions concentrating mass $1-\varepsilon$ on a single state, continuity of utilities in the initial distribution lets us pass to the limit $\varepsilon \to 0$, yielding an MPE. This generalizes the existence of MPE in Markov games (see e.g. \cite{deng-et-al23}) via a different proof technique. 

\noindent\textbf{PG theorem and model-free algorithm.} Despite non-additivity of utilities, we exploit the dynamic programming structure of occupancy measures to establish a PG theorem for $N$-player GUMGs. This leads to a model-free PG algorithm which does not require full knowledge of the cMG transition model (in contrast with \citealt{gemp-et-al25convex-mgs}). 

\noindent\textbf{Sample-efficient learning in potential GUMGs.} 
In the common interest setting, our analysis relies on showing smoothness of utilities w.r.t. policies in GUMGs, proving convergence to approximate stationary points using stochastic gradient ascent arguments and exploiting our gradient domination property to show convergence to approximate NE. Our sampling mechanism in the generative model setting differs from \cite{zhang-et-al24tac} as we use minibatch trajectory sampling \textit{without} $\alpha$-greedy policies, leading to an improved sample complexity of~$\tilde{\mathcal{O}}(\varepsilon^{-4})$ instead of $\tilde{\mathcal{O}}(\varepsilon^{-6})\,.$ In the on-policy setting, our proof technique relies on an analysis of stochastic projected gradient ascent with minibatch for smooth objectives, differing from the Moreau-envelope-based analysis used for the special case of MPGs in \citep{leonardos-et-al22iclr}. 

\subsection{Related Work}
\label{sec:related-work}

Recent years have witnessed active research in theoretical MARL. In particular, an important body of work has focused on 
learning approximate NE in structured classes of Markov games such as zero-sum Markov games and Markov potential games. Our work extends these research efforts to cMGs and beyond.  
Recently, \cite{gemp-et-al25convex-mgs} introduced cMGs, proved existence of NE and proposed to approximate them using a \textit{biased} entropic regularization approach under \textit{full knowledge} of the transition model. In this work, we design a \textit{model-free} PG algorithm for GUMGs and analyze its iteration and \textit{sample complexity}. In a follow-up work, \cite{kalogiannis-et-al25zs-mgs} studied \textit{2-player zero-sum cMGs}. In contrast, we characterize NE in any $N$-player general-sum GUMG (with cMGs as particular cases). Moreover, we study the common-interest setting beyond the competitive zero-sum two-player setting. See sections~\ref{sec:structure-general-cMGs} and~\ref{sec:pg-algo} for more details comparing our theoretical and algorithmic contributions to prior work. 
In the single-agent setting ($N=1$), convex Markov games reduce to convex RL which is an active research topic \citep{hazan-et-al19,zahavy-et-al21,zhang-et-al20variational,zhang-et-al21rlgu,agarwal-et-al22,geist-et-al22,mutti-et-al23convexRL-jmlr,barakat-et-al23rl-gen-ut,marin-et-al24metacurl,desanti-et-al24,barakat-et-al25rlgu}. A few works consider MARL with general utilities as a \textit{decentralized policy optimization} problem \citep{zhang-et-al22marl-gu-decentralized,ying-et-al23scalable-safe-marl-rlgu}. Our game-theoretic setting is fundamentally different. 
See Appendix~\ref{sec:extended-related-work} for an extended literature review. 

\section{\normalsize{FROM CLASSICAL TO GENERAL UTILITY MARKOV GAMES}} 

In this section, we introduce useful notation, Markov games and cMGs to motivate our class of GUMGs. 

\noindent\textbf{Notation.} 
We denote by $[N]$ the set of integers~$\{1, \cdots, N\}\,.$ 
For any vector $v \in \mathbb{R}^n$ where $n$ is any integer larger than one, we denote its infinity norm by~$\|v\|_{\infty} := \max_{i=1, \cdots, N} |v_i|\,.$  
For any finite set~$\mathcal{X}$ with cardinality~$X = |\mathcal{X}|,$ we denote by~$\Delta(\mathcal{X})$ the set of (discrete) probability distributions on~$\mathcal{X}\,.$ This set can be simply identified with the standard simplex, i.e. the set $\{\pi \in \R^{X}: \pi_i \geq 0\,, \forall i \in [X]\,, \sum_{i=1}^{X} \pi_i = 1 \}\,.$

\subsection{Markov Games and Notation}
\label{sec:markov-games-notation}

A Markov game is defined via a $6$-tuple: 
\begin{equation}
\label{eq:markov-game}
\mathcal{G} = ( \mathcal N, \mathcal{S}, \{ \mathcal A_i \}_{i \in \mathcal N}, P, \{r_i\}_{i \in \mathcal N}, \mu, \gamma)\,,
\end{equation}
where $\mathcal{N} := \{1, \cdots, N\}$ is a finite set of $N$ agents, $\mathcal{S}$ is a finite set of states, $\mathcal{A}_i$ is a finite set of actions of cardinality $|\mathcal{A}_i|$ for all~$i \in \mathcal N$. The joint action space is denoted by~$\mathcal{A} := \prod_{i \in \mathcal N} \mathcal{A}_i\,.$ Moreover, we denote by $P: \mathcal{S}\times\mathcal{A} \to \Delta(\mathcal{S})$ the probability transition kernel: For any state~$s \in \mathcal{S}$ and any joint action~$a \in \mathcal{A}$, the game transitions from state~$s$ to a state~$s' \in \mathcal{S}$ with probability~$P(s'|s, a)\,.$ 
The initial state probability distribution is denoted by~$\mu \in \Delta(\mathcal{S})$ and $\gamma \in (0,1)$ is a discount factor. Each agent~$i \in \mathcal{N}$ has their own reward function $r_i: \cS \times \cA \to \R$ assigning a reward $r_i(s,a)$ to any state~$s \in \cS$ and \textit{joint} action~$a \in \cA\,.$   

\noindent\textbf{Policies.} Each agent~$i \in \mathcal{N}$ chooses their actions according to a randomized stationary policy denoted by $\pi_i \in \Pi_i:=\Delta \left( \mathcal A_i \right)^{\mathcal S}$. 
The set of joint policies~$\pi = \left( \pi_i \right)_{i \in \mathcal N}$ is denoted by~$\Pi := \times_{i \in \mathcal{N}} \Pi_i$ and we further use the notation $\pi_{-i}=\left( \pi_j \right)_{j \in \mathcal{N} \setminus \{i\}} \in \Pi_{-i} := \times_{j \in \mathcal{N} \setminus \{i\}} \Pi_j$ for joint policies of all agents other than~$i$. 
In this work, we will also be considering $\alpha$-greedy policies which guarantee a level of exploration by ensuring that each action is selected with positive probability. Define for each agent $i \in \mathcal{N}$ the set of $\alpha$-greedy policies $\Pi_i^{\alpha}:= \{\pi_i^{\alpha} \in \Pi_i: \forall s \in \mathcal{S}, a_i \in \mathcal{A}_i, \pi_i^{\alpha}(a_i|s) := (1-\alpha) \pi_i(a_i|s) + \alpha/|\mathcal{A}_i|, \pi_i \in \Pi_i \}$ for $\alpha > 0\,.$ The set of joint $\alpha$-greedy policies is denoted by $\Pi^{\alpha} := \times_{i \in \mathcal{N}} \Pi_i^{\alpha}\,.$ 

At each time step~$t \in \bN$ in a state~$s_t \in \mathcal{S}$, each agent~$i \in \mathcal{N}$ chooses an action~$a_{i,t} \in \mathcal{A}_i$ with probability~$\pi_i(a_{i,t}|s_t)$ and then the environment transitions to a state~$s_{t+1}\in \mathcal{S}$ with probability~$P(s_{t+1}|s_t, a_t)$ where~$a_t = (a_{i,t})_{i \in \mathcal{N}}$ is the joint action.   
We denote by~$\bP_{\mu,\pi}$ the probability distribution of the Markov chain~$(s_t,a_t)_{t \in \mathbb{N}}$ induced by the joint policy~$\pi$ with initial state distribution~$\mu$.  
We use the notation~$\bE_{\mu,\pi}$ (or often simply $\bE$) for the associated expectation.

\noindent\textbf{Occupancy measures.} 
We define for any policy~$\pi \in \Pi$ the (normalized) state and state-joint-action occupancy measures~$d_{\mu}^{\pi} \in \Delta(\mathcal{S}), \lambda_{\mu}^{\pi} \in \Delta(\mathcal{S} \times \mathcal{A})$ which capture the frequency of state-action pairs visited:
\begin{equation}
\label{eq:s-a-occup-measure}
\lambda_{\mu}^{\pi}(s,a) \eqdef (1-\gamma) \sum_{t=0}^{+\infty} \gamma^t \bP_{\mu,\pi}(s_t = s, a_t = a)\,, 
\end{equation}
and~$d_{\mu}^{\pi}(s) = \sum_{a \in \cA} \lambda_{\mu}^{\pi}(s,a)\,.$ 
Note that~$\lambda_{\mu}^{\pi}$ will also be seen as a vector of the Euclidean space~$\R^{|\mathcal{S}| \cdot |\mathcal{A}|}$ with a slight abuse of notation. 
We also define for each agent~$i \in \mathcal{N}$ the marginal occupancy measure induced by the joint policy~$\pi$ as follows for every $s \in \mathcal{S}, a_i \in \mathcal{A}_i,$ 
\begin{equation}
\label{eq:marginal-occup-measure}
\lambda_{\mu,i}^{\pi}(s,a_i) \eqdef (1-\gamma) \sum_{t=0}^{+\infty} \gamma^t \bP_{\mu,\pi}(s_t = s, a_{i,t} = a_i)\,. 
\end{equation}
Note that $\lambda_{\mu,i}^{\pi}(s,a_i) = \sum_{a_{-i} \in \cA_{-i}} \lambda_{\mu}^{\pi}(s,a_i, a_{-i})\,.$ 

\noindent\textbf{Value functions.} Finally, we introduce value functions which inform on the quality of a given initial state(-action) in terms of cumulative discounted rewards under a given policy. For any standard MDP $\mathcal{M} := (\mathcal{S}, \mathcal{A}, P, r, \mu, \gamma)$ where $r \in \mathcal{R}^{|\mathcal{S}| \cdot |\mathcal{A}|}$ refers to a reward function (i.e. $N=1$ in \eqref{eq:markov-game}), we define the (action)-value functions for every $(s,a) \in \mathcal{S} \times \mathcal{A},$

$Q_{s,a}^{\pi}(r) := \mathbb{E}_{\rho, \pi}\left[ \sum_{t=0}^{+\infty} \gamma^t r(s_t, a_t) | s_0 = s, a_0 = a \right]\,,$
and $V_{\rho}^{\pi}(r) = \sum_{a \in \mathcal{A}} \pi(a|s) Q_{s,a}^{\pi}(r)\,.$

\subsection{Convex Markov Games}
\label{sec:cmg}

As introduced by \cite{gemp-et-al25convex-mgs}, a convex Markov Game (cMG) is a $6$-tuple $\mathcal{G} =( \mathcal N, \mathcal{S}, \{ \mathcal A_i \}_{i \in \mathcal N}, P, \{F_i\}_{i \in \mathcal N}, \mu, \gamma)\,,$ 
where $(\mathcal N, \mathcal{S}, \{ \mathcal A_i \}_{i \in \mathcal N}, P, \mu, \gamma)$ is defined as in standard Markov games (see section~\ref{sec:markov-games-notation} above). 
However, in contrast to Markov games, utilities of agents are not restricted to be defined as cumulative discounted rewards using some given fixed reward functions.  
Each agent $i \in \mathcal N$ has a utility function $F_i: \Lambda_i \times \Pi_{-i} \to \R$\, where $\Lambda_i := \Delta(\mathcal{S} \times \mathcal{A}_i)$. The utility of agent~$i$ in the cMG under a joint policy~$\pi \in \Pi$ is then given by: 
\begin{equation}
\label{eq:def-cMG}
u_{\mu,i}(\pi_i, \pi_{-i}) = F_i(\lambda_{\mu,i}^{\pi}, \pi_{-i})\,,
\end{equation}
where the marginal occupancy measure~$\lambda_{\mu,i}^{\pi}$ is defined in~\eqref{eq:marginal-occup-measure}. Each utility function is supposed to be continuous and \textit{concave} in player $i$’s occupancy measure, and continuous in each policy $\pi_j$ for $j \neq i\,.$ 
Notice that the utility of agent~$i$ \textit{does not} depend explicitly on the individual occupancy measures of \textit{other agents}. Hence, coupling between occupancy marginals of agents is excluded. 
This motivates our more general definition which captures such couplings.

\subsection{General Utility Markov Games} 

In this section, we introduce our new class of games which generalizes the class of cMGs. 

\begin{definition}[General Utility Markov Game]
\label{def:gumg}
A General Utility Markov Game is defined via a 6-tuple 
$( \mathcal N, \mathcal{S}, \{ \mathcal A_i \}_{i \in \mathcal N}, P, \{F_i\}_{i \in \mathcal N}, \mu, \gamma)$ where $( \mathcal N, \mathcal{S}, \{ \mathcal A_i \}_{i \in \mathcal N}, P, \mu, \gamma)$ are defined as in section~\ref{sec:cmg} and each real-valued function $F_i$ is defined over the product space $\Lambda \times \Pi_{-i}$ where $\Lambda := \prod_{i=1}^N \Lambda_i$ and  $\Lambda_i := \Delta(\cS \times \cA_i)$. 
The utility function~$u_i : \Pi \to \mathbb{R}$ of agent~$i \in \mathcal{N}$ is defined for every policy $\pi = (\pi_1, \cdots, \pi_N) \in \Pi$ by
\begin{equation}
\label{eq:def-u_i}
u_{\mu,i}(\pi_i, \pi_{-i}) := F_i(\lambda_{\mu,1}^{\pi}, \cdots, \lambda_{\mu,N}^{\pi}, \pi_{-i})\,. 
\end{equation}
\end{definition}

Throughout this work we will suppose that $F_i$ is differentiable and this implies that $u_i$ is also differentiable.
We will also make the following important assumption. 
\begin{assumption}[Concavity] 
\label{as:concavity}
For all $i \in \mathcal{N}$, the function $F_i$ is jointly concave in $(\lambda_1, \cdots, \lambda_N)$. 
\end{assumption}

\noindent\textbf{Why GUMGs?} Assumption~\ref{as:concavity} captures cMGs, and therefore Markov games and convex RL ($N=1$) as well.  Compared to cMGs, GUMGs capture utilities that are concave in per-agent occupancies and \textit{aggregates} thereof. 
See examples in the next section. 

\begin{remark}
\label{rem:alternative-def}
Another possible and simple definition using the occupancy measure of the joint policy consists of defining $u_{\mu,i}(\pi) = F_i(\lambda_{\mu}^{\pi})$ for any policy $\pi \in \Pi.$ Our results in section~\ref{sec:structure-general-cMGs} will also hold for this definition. However, this is a particular case of cMGs and leads to the need of estimating the joint occupancy measure which cannot be directly estimated by each agent independently (in contrast to marginal occupancies). 
\end{remark}

\subsection{Examples} 
\label{subsec:examples}

Consider the following composite example defined for any $\lambda = (\lambda_1, \cdots, \lambda_N) \in \Lambda, \pi = (\pi_i, \pi_{-i}) \in \Pi,$ by: 
\begin{multline}
\label{eq:ex-composite}
F_i(\lambda_1, \cdots, \lambda_N, \pi_{-i}) = 
 \underbrace{ - \alpha_i \text{KL}(\lambda_i || q_i)}_{\text{imitation}} \quad \underbrace{- \beta_i \text{KL}(\lambda_i || \bar{\lambda}_N)}_{\text{consensus/diversity}}\\
+ \underbrace{\gamma_i h_i \left(\sum_{j=1}^N W_{i,j} \lambda_j \right)}_{\text{team aggregate}} + \underbrace{g_i(\pi_{-i})}_{\text{diversity penalty}}\,,
\end{multline}
where $q_i \in \Lambda_i, i \in \mathcal{N}$ are fixed probability distributions, $\alpha_i, \beta_i, \gamma_i$ are nonnegative reals, $W \in \mathbb{R}^{N \times N}$ is a stochastic matrix of weights,
and $\bar{\lambda}_N = \sum_{i=1}^N w_i \lambda_i$ with $w_i \geq 0, \sum_{i=1}^N w_i = 1\,.$ Here, $h_i$ and $g_i$ are concave functions. Note that $F_i$ satisfies Assumption~\ref{as:concavity} since it is the sum of jointly concave functions, noticing that $\text{KL}$ divergence is jointly convex and the composition of a convex function with a linear one is also convex.  

Notice that the above composite example for agent $i$'s utility shows dependence on the occupancy measures of \textit{all} agents via the aggregated terms $\bar{\lambda}_N$ and $\sum_{j=1}^N W_{i,j} \lambda_j$. This coupling between agents is \textit{not} captured by the definition of cMGs proposed in \cite{gemp-et-al25convex-mgs}, see \eqref{eq:def-cMG}. We now elaborate on special cases of our composite example~\eqref{eq:ex-composite}: 

\noindent\textbf{a. Multi-agent imitation learning} ($\alpha_i > 0,\beta_i = \gamma_i = 0$). Each agent matches its own occupancy to its expert occupancy~$q_i$. See e.g. \cite{song-et-al18}. 

\noindent\textbf{b. Consensus and diversity}  ($\beta_i > 0,\alpha_i = \gamma_i = 0$).    
The second term in \eqref{eq:ex-composite} couples marginals via~$\bar{\lambda}_N$ which can be computed from broadcast aggregates or neighbor averages without reconstructing the occupancy of the joint policy. Moreover, the term~$g_i(\pi_{-i})$ can be interpreted as a policy diversity penalty. See e.g. \cite{li-et-al21diversity} for MARL applications using information-theoretical regularization for behavioral diversity. 

\noindent\textbf{c. Population-level coordination} ($\gamma_i > 0,\alpha_i = \beta_i = 0$).   
Using per-agent marginals, utilities depend on the population behavior or neighbor average. 
See mean-field MARL \citep{yang-et-al18mean-field}. 

\noindent\textbf{d. Team coverage and distribution matching} 
($F_i(\lambda_{1:N},\pi_{-i}) = - \text{KL}(\bar{\lambda}_N || \lambda_{\text{target}})$). This utility encourages matching team visitation with a desired target distribution~$\lambda_{\text{target}}$. Important examples include patrolling, area coverage, search and target localization, see e.g. \cite{mavrommati-et-al17} in robotics.

\noindent\textbf{e. Collective exploration} ($\gamma_i > 0,\alpha_i = \beta_i = 0$). When $h_i$ is the negative entropy, utility maximization enhances exploration.  
See e.g. \cite{kim-sung23,zamboni-et-al25}. 

\noindent\textbf{Beyond GUMGs?}
There are individual objectives that do not satisfy our \textit{joint} concavity assumption (e.g. $F_1(\lambda_1, \lambda_2, \pi_{-i}) = \lambda_1 \lambda_2$ when $N=2$). 
We leave the question of whether our theory can be extended to even more general functions~$\{F_i\}_{i \in \mathcal{N}}$ to future work. 

\section{EXISTENCE AND STRUCTURE OF NASH POLICIES IN GUMGs}
\label{sec:structure-general-cMGs}

Recall first the definition of an (approximate)-NE: a joint policy where no agent can improve its utility by deviating unilaterally. For any $\varepsilon \geq 0,$ a policy $\pi^{\star} \in \Pi$ is an $\varepsilon$-approximate NE of a GUMG if for all~$i \in \mathcal{N}$, all policies $\pi_i' \in \Pi_i$, $u_{\mu,i}(\pi_i', \pi_{-i}^{\star}) - u_{\mu,i}(\pi_i^{\star}, \pi_{-i}^{\star}) \leq \varepsilon\,.$ When $\varepsilon = 0,$ the policy~$\pi^{\star}$ is a NE policy. 
A Markov perfect equilibrium (MPE) is a Markov policy that constitutes a NE for every initial state $s \in \mathcal{S}$, i.e., the definition of NE holds for any $\mu = \delta_s$ for MPE (where $\delta_s$ is the Dirac measure at state $s$) rather than only for the fixed initial state distribution $\mu$ defining the GUMG\footnote{Note that utilities are indexed by the initial state distribution to highlight this dependence.}. This means that the inequalities in the definition of NE hold for every distribution $\mu = \delta_s, s \in \mathcal{S}$ (and hence for any initial state distribution as well). See Appendix~\ref{app:ne-vs-mpe} for more details regarding the difference between NE and MPE and its relevance. 

Our goal in this section is to show that in GUMGs, NE exist and admit a first-order characterization under Assumption~\ref{as:concavity}, revealing their connection to projected pseudo-gradient dynamics. First, we make a standard exploration assumption to guarantee that all states are visited with positive probability in the long run.

\begin{assumption}
\label{as:exploration}
For all $\pi \in \Pi, s \in \mathcal{S}$, $d_{\mu}^{\pi}(s) > 0\,.$ 
\end{assumption}
Note that Assumption~\ref{as:exploration} is satisfied for instance when the initial state distribution~$\mu$ has full support since we always have $d_{\mu}^{\pi}(s) \geq (1-\gamma)\mu(s)$ for all states~$s \in \cS$. 

The utility of each agent in GUMGs is non-concave in their policy in general, as this is already the case in standard single-agent RL (see e.g. Lemma~1 in \cite{agarwal-et-al21jmlr}).  
Nevertheless, the following crucial result establishes a 
gradient domination property which holds per-agent. This structural property of GUMGs relates utility change under unilateral policy deviation to individual policy gradients. 

\begin{proposition}[Agentwise Gradient Domination]
\label{prop:agentwise-grad-dom} 
Let Assumptions~\ref{as:concavity} and~\ref{as:exploration} hold.
For all $i \in \mathcal{N}$, all policies $\pi = (\pi_i, \pi_{-i}) \in \Pi$, $\pi_i' \in \Pi_i$, we have for any distribution~$\mu, \rho \in \Delta(\cS)$ satisfying Assumption~\ref{as:exploration}, 
\begin{equation*}
u_{\mu,i}(\pi_i', \pi_{-i}) - u_{\mu,i}(\pi_i, \pi_{-i}) \leq C_{\mathcal{G}} \cdot \max_{\tilde{\pi}_i \in \Pi_i} \langle \nabla_{\pi_i} u_{\rho,i}(\pi), \tilde{\pi}_i - \pi_i \rangle
\end{equation*}
where the minimax distribution mismatch coefficient: 
$C_{\mathcal{G}} := \max_{i \in \mathcal{N}} \max_{\pi \in \Pi} \min_{\pi_i^{\star} \in \Pi_i^{\star}(\pi_{-i})} \left\| \frac{d_{\mu}^{\pi_i^{\star}, \pi_{-i}}}{d_{\rho}^{\pi_i, \pi_{-i}}} \right\|_{\infty}\,,$ 
where $\Pi_i^{\star}(\pi_{-i})$ denotes the set of best response policies for agent~$i$, is finite. 
\end{proposition}

Prop.~\ref{prop:agentwise-grad-dom} generalizes the per-agent gradient domination property shown in \cite{daskalakis-foster-golowich20} for 2-player zero-sum Markov games, in \cite{leonardos-et-al22iclr} for Markov potential games and in \cite{zhang-et-al24tac,giannou-et-al22neurips} for general-sum Markov games.

\begin{remark} The mismatch coefficient~$C_{\mathcal{G}}$ characterizes hardness of state exploration. This constant can be loosely bounded by $1/((1-\gamma)\min_{s \in \cS} \rho(s))$ when $\rho$ has full support, since for any policy~$\pi \in \Pi, (1-\gamma) \rho(s) \leq d_{\rho}^{\pi}(s) \leq 1$ for all~$s \in \cS.$ Note also, as mentioned in \citet[Proposition 1]{daskalakis-foster-golowich20}, that finiteness of the minimax coefficient~$C_{\mathcal{G}}$ is in general weaker than finiteness of the concentrability coefficient~$ \max_{\pi \in \Pi} \left\| \frac{d_{\mu}^{\pi}}{\mu} \right\|_{\infty}$ which has been used in prior work in Markov games \citep{leonardos-et-al22iclr,ding-et-al22icml,zhang-et-al24tac,alatur-barakat-he24cdc}. 
\end{remark}

Recently, a similar result was proved in \cite[Theorem~1]{barakat-et-al25rlgu} for the \textit{single-agent} case and in \cite[Lemma 2.5]{kalogiannis-et-al25zs-mgs} for \textit{2-player} convex Markov games. Note that (i) our result holds for any number $N$ of players and for any GUMG beyond $2$-player zero-sum cMGs, and (ii) our bound is tighter as it uses a distribution mismatch coefficient which matches the two-player Markov game gradient domination result of \cite{daskalakis-foster-golowich20}. In particular, our proof technique differs from \cite{kalogiannis-et-al25zs-mgs} which rather relies on hidden convexity \citep{fatkhullin-et-al25}. Instead, we leverage the single agent gradient domination RL result of \cite{agarwal-et-al21jmlr} in our setting and make use of our concavity assumption. The first step in the proof uses concavity to relate utility deviation to standard single-agent RL value function deviation using pseudo-rewards (by defining a suitable MDP to reduce general utility to standard RL). The second step applies single-agent gradient domination for standard RL when the policies of other players' but $i$ are fixed, and uses our policy gradient expression to recover the first-order variation in the desired inequality. See proof in Appendix~\ref{sec:proof-prop-agentwise-grad-dom}. 

In what follows, we use the agent-wise gradient domination property of utilities (Proposition~\ref{prop:agentwise-grad-dom}) to relate NE to fixed points of the projected pseudo-gradient dynamics. First, we introduce the \textit{pseudo-gradient} vector field notation for every joint policy~$\pi \in \Pi$: 
\begin{equation}
\label{eq:pseudo-grad}
v_i(\pi) :=  \nabla_{\pi_i} u_{\mu,i}(\pi)   \,,\quad  v(\pi) =  (v_i(\pi))_{i\in \mathcal{N}}\,.
\end{equation}
Note that the vector field~$v$ is not necessarily a gradient field in general, i.e. the existence of a function~$f$ s.t. $v = \nabla f$ is not guaranteed in general. Indeed, each~$v_i$ is the gradient of a different function~$u_i\,.$

It is well-known that Nash equilibria are first-order stationary. However, the converse implication is not true in general. Thanks to agent-wise gradient domination, the next result shows that Nash policies are equivalent to first-order stationary points in GUMGs.  

\begin{proposition}[First-order Stationary policies are Nash policies] 
\label{prop:fos-nash-equivalence}
Let Assumptions~\ref{as:concavity} and~\ref{as:exploration} hold. 
A (joint) policy $\pi^{\star} \in \Pi$ is a Nash policy in a GUMG if and only if it is first-order stationary, i.e. it satisfies: 
\begin{equation}
\label{eq:def-fos}
\langle v(\pi^{\star}), \pi - \pi^{\star} \rangle \leq 0, \quad \forall \pi \in \Pi\,. 
\tag{FOS}
\end{equation}
Moreover, for any $\epsilon > 0,$ if a policy $\pi^{\star}$ is a $\varepsilon$-FOS policy, i.e. it satisfies: 
$\langle v(\pi^{\star}), \pi - \pi^{\star} \rangle \leq \varepsilon, \, \forall \pi \in \Pi\,,$ 
then $\pi^{\star}$ is also a $(C_{\mathcal{G}} \,\varepsilon)$-Nash policy, where  $C_{\mathcal{G}}$ is the distribution mismatch coefficient defined in Prop.~\ref{prop:agentwise-grad-dom}. 
\end{proposition}

\ref{eq:def-fos} points coincide with fixed points of the projected pseudo-gradient dynamics. Indeed, a policy~$\pi^{\star}$ is \ref{eq:def-fos} if and only if it is a fixed point of the projected gradient dynamics, i.e.,   
$\pi^{\star} = \text{Proj}_{\Pi}(\pi^{\star} + \eta\, v(\pi^{\star})),\forall \eta > 0$ (see Lemma~\ref{prop:characterization-fos-fixed-point} in Appendix~\ref{app:proof-lemma-fos-fixed-pts} for a proof of this straightforward result).
Combining this result with Prop.~\ref{prop:fos-nash-equivalence}, we obtain a characterization of Nash policies in cMGs.

\begin{theorem} 
\label{prop:equiv-ne-fixed-point-projgrad}
Let Assumptions~\ref{as:concavity} and~\ref{as:exploration} hold.  
A (joint) policy $\pi^{\star} \in \Pi$ is a NE in a GUMG if and only if it is a fixed point of the projected gradient dynamics, i.e.: 
\begin{equation}
\label{eq:fixed-point-proj-grad}
\pi^{\star} = \text{Proj}_{\Pi}(\pi^{\star} + \eta\, v(\pi^{\star})), \quad \forall \eta > 0\,,
\end{equation}
or equivalently, using the product structure of the joint policy space and the properties of the projection, 
$\pi_i^{\star} = \text{Proj}_{\Pi_i}(\pi_i^{\star} + \eta\, v_i(\pi^{\star})), \forall i \in \mathcal{N},\, \forall \eta > 0\,.$
\end{theorem} 

\begin{corollary}
\label{cor:existence-NE-cMG}
Under Assumption~\ref{as:concavity}, there exists at least one Nash equilibrium in GUMGs as defined in Definition~\ref{def:gumg}. Moreover, there exists at least one MPE.
\end{corollary}
\begin{proof}
Fixed points of the projected gradient dynamics described in \eqref{eq:fixed-point-proj-grad} exist by Brouwer's fixed point theorem under continuity of the mapping~$v$ (which follows from Proposition~\ref{prop:smoothness} below) observing that the set of policies~$\Pi$ is compact. 
It follows from the equivalence of Proposition~\ref{prop:equiv-ne-fixed-point-projgrad} that Nash policies do exist.   
\end{proof}

In the special case of Markov games, corollary~\ref{cor:existence-NE-cMG} provides an alternative proof to the original proof of \cite{fink64} which uses Kakutani's fixed point theorem (which is for set-valued functions). Our first-order characterization of NE in GUMGs allows us to bypass the need for a set-valued function theorem and directly apply Brouwer's fixed point theorem. Existence of NE for cMGs has been recently proved in \cite{gemp-et-al25convex-mgs} via \cite{debreu52social-existence-eq-thm}'s theorem using contractible sets (which are not necessarily convex).  
Thm.~\ref{prop:equiv-ne-fixed-point-projgrad} shows that it is sufficient to find a fixed point of the projected pseudo-gradient dynamics to find a NE. This motivates the design of policy gradient algorithms for GUMGs, which we discuss in the next section. 

\section{POLICY GRADIENT ALGORITHM FOR GUMGs} 
\label{sec:pg-algo}

In this section, we design a policy gradient algorithm for GUMGs. We first establish a 
useful multi-agent policy gradient theorem. 
In GUMGs, the gradient of the utility of each player w.r.t. their own policy is given by the following result. 

\begin{proposition}[Multi-Agent PG Theorem]
\label{prop:individual-pg}
    For any agent~$i \in \mathcal{N}$, any joint policy~$\pi = (\pi_1, \cdots, \pi_N) \in \Pi$, 
    \vspace{-2mm}
    \begin{enumerate}[label=(\roman*),leftmargin=*]
        \item \label{prop-i:individual-pg} $\nabla_{\pi_i} u_{\mu,i}(\pi) = \mathbb{E}_{\mu, \pi}\left[ \sum_{t=0}^{+\infty} \gamma^t R_i(s_t, a_t)\, \psi_{\pi_i}^t \right]\,,$
        where $R_i(s_t, a_t) := \sum_{j=1}^N \bar{r}_{ij}^{\pi}(s_t, a_{j,t})$,\\
        $\bar{r}_{ij}^{\pi} := \nabla_{\lambda_j} F_i(\lambda_{\mu,1}^{\pi}, \cdots, \lambda_{\mu,N}^{\pi}, \pi_{-i})$ for $j \in \mathcal{N}$,\\
        $\psi_{\pi_i}^t := \sum_{t' = 0}^t \nabla_{\pi_i} \log \pi_i(a_{i,t'}|s_{t'})\,.$
        \item \label{prop-ii:individual-pg} For any state~$s \in \mathcal{S}$ and any action $a_i \in \mathcal{A}_i$, 
        \begin{align}
        \frac{\partial u_{\mu,i}(\pi)}{\partial \pi_i(a_i|s)} &= \frac{1}{1-\gamma} d_{\mu}^{\pi}(s) \sum_{j=1}^N \bar{Q}_{s,a_i}^{\pi}(\bar{r}_{i,j}^{\pi})\,,\\
        \bar{Q}_{s,a_i}^{\pi}(\bar{r}_{i,j}^{\pi}) &:= \sum_{a_{-i} \in \mathcal{A}_{-i}} \pi_{-i}(a_{-i}|s) Q_{s,a}^{\pi}(\bar{r}_{i,j}^{\pi}) \nonumber\,,
        \end{align}
        where $a = (a_i, a_{-i})\,.$ 
    \end{enumerate} 
\end{proposition}
The \textit{pseudo-reward} quantity~$\bar{r}_{i,j}^{\pi}$ is an \textit{implicit policy-dependent} reward (in contrast to rewards in classical RL or Markov games). 
For instance, in pure exploration where $F_i$ is the entropy function, the goal is to explore a state space without reward signals. 

Proposition~\ref{prop:individual-pg} recovers as particular cases recent existing multi-agent policy gradient theorems for Markov games (see e.g. \cite{leonardos-et-al22iclr}, \cite{zhang-et-al24tac}), for standard single-agent RL ($N=1$) \citep{sutton-et-al99PG,williams92reinforce} as well as for (single-agent) RL with general utilities \citep{zhang-et-al21rlgu}. As for the proof, we first use the chain rule  to differentiate our GUMG objective and then use the standard policy gradient using pseudo-rewards and the corresponding average Q functions. The proof is inspired from prior work in standard Markov games and single-agent RL with general utilities.

\noindent\textbf{Algorithm design.} 
Our results in section~\ref{sec:structure-general-cMGs} (Thm.~\ref{prop:equiv-ne-fixed-point-projgrad}) suggest to perform projected pseudo-gradient ascent:
\begin{equation}
\pi_{i}^{t+1} = \text{Proj}_{\Pi_i^{\alpha}}\left(\pi_{i}^{t} + \eta  \nabla_{\pi_i} u_{\mu,i}(\pi^t) \right)\,, 
\end{equation}
for each time step~$t$, using a positive step size~$\eta$. In light of Thm.~\ref{prop:individual-pg}, we replace exact gradients by stochastic policy gradients computed using sampled finite-length trajectories as described in Algorithm~\ref{algo}:
\begin{equation}
\label{eq:stochastic-pg}
\hat{\nabla}_{\pi_i} u_{\mu,i}(\pi) := \sum_{t=0}^{H-1} \gamma^t \hat{R}_i(s_t,a_t) \psi_{\pi_i}^t\,,
\end{equation}
where $\hat{R}_i(s_t,a_t) := \sum_{j=1}^N \hat{r}_{i,j}(s_t, a_{t,j})$ and~$\psi_{\pi_i}^t$ is defined in Prop.~\ref{prop:individual-pg}. 
Note that marginal occupancy measures are estimated by $\hat{\lambda}_{\mu,i}^{\pi}(s,a_i) = \hat{d}_{\mu}^{\pi}(s) \pi_i(a_i|s)$ where state occupancy measure estimates are given by: 
\begin{equation}
\label{eq:occup-estimate}
\hat{d}_{\mu}^{\pi}(s) =  \frac{1}{M} \sum_{k=1}^M \sum_{t=0}^{H-1} \gamma^t \mathds{1}_{\{s_t^{(k)} = s\}}\,. 
\end{equation} 

\begin{algorithm}[t]
   \caption{PG Algorithm for GUMGs (for $i \in \mathcal{N}$)}
   \label{algo}
\begin{algorithmic}[1]
   \STATE {\bfseries Input:} $\pi_i^0 \in \Pi_i, T, N \geq 1, \eta > 0,  H\,.$ 
   \FOR{$t=0, \ldots, T -  1$}
        \STATE Sample~$M$ trajectories~$(\tau_{i,k})_{1 \leq k \leq M}$ of length~$H$: $\tau_{i,k} = \{(s_t^{(k)}, a_{i,t}^{(k)})_{0 \leq t \leq H-1}\}, a_{i,t} \sim \pi_i^t(\cdot|s_t)\,.$ 
        \STATEx {\color{blue}\texttt{//Pseudo-reward estimation}}
        \STATE Compute $\hat{\lambda}_{\mu,i}^{\pi^t}$ using Eq.~\eqref{eq:occup-estimate}. 
        \STATE Send $\hat{\lambda}_{\mu,i}^{\pi^t}$ to agents $j \in \mathcal{N}, j \neq i\,.$
        \STATE Receive $\hat{\lambda}_{\mu,j}^{\pi^t}, j \in \mathcal{N}, j \neq i\,.$ 
        \STATE Compute $\hat{r}_{i,j} = \nabla_{\lambda_j} F_i(\hat{\lambda}_{\mu,1:N}^{\pi^t},\pi_{-i}^t), j \in \mathcal{N}$ 
        \STATEx {\color{blue}\texttt{//Policy gradient estimation}}
        \STATE Compute PG estimate $\hat{\nabla}_{\pi_i} u_{\mu,i}(\pi^t)$ using~\eqref{eq:stochastic-pg} (on-policy) or \eqref{eq:stochastic-pg-gen-model} (under a generative model). 
        \STATEx {\color{blue}\texttt{//Policy update}}
        \STATE $\pi_{i}^{t+1} = \text{Proj}_{\Pi_i^{\alpha}}\left(\pi_{i}^{t} + \eta  \hat{\nabla}_{\pi_i} u_{\mu,i}(\pi^t) \right)$
    \ENDFOR
\end{algorithmic}
\end{algorithm}

\noindent\textbf{Communication protocol.} In a GUMG, to implement the policy gradient Algorithm~\ref{algo}, each agent needs to compute its pseudo-reward $(\bar{r}_{ij}^{\pi} = \nabla_{\lambda_j} F_i(\lambda_{\mu,1}^{\pi}, \cdots, \lambda_{\mu,N}^{\pi}, \pi_{-i}))$ which depends on the (marginal) occupancy measures of all agents. A sufficient and practical protocol is for each agent to broadcast their estimated occupancy measure and  their policy to the other agents at the beginning of each interaction round (steps 5–6 of Algorithm~\ref{algo}). Each agent can then compute pseudo-rewards locally and update their policy independently. 
Note that policies of other agents can also be inferred from state-action occupancy measures via marginalization given knowledge of their action spaces. 
The communication requirement stems directly from the fact that individual GUMG objectives depend on other agents’ occupancies and policies in our general model. In particular, communication is not required in special cases such as standard Markov games or settings where each agent’s utility depends only on its own occupancy measure. We also emphasize that policies need not be communicated if the utilities $F_i$ do not have a dependence on $\pi_{-i}$ via their last argument (while they can still depend on $\pi_{-i}$ via the marginal occupancy measures).

We now highlight some features of our algorithm:

\noindent\textbf{Model-free algorithm.} Algorithm~\ref{algo} is model-free as it only relies on sampled trajectories to estimate policy gradients. In contrast, the algorithm proposed in \cite{gemp-et-al25convex-mgs} for approximate NE computation in cMGs assumes full knowledge of the transition model. Our algorithm design is rather inspired by PG methods for RL with general utilities \citep{zhang-et-al21rlgu}. Compared to Markov games, an important difference is the need to estimate marginal occupancies~$\lambda_{\mu,j}^{\pi}$ to estimate the pseudo-rewards (which only need to be estimated at sampled visited state-action pairs).  

\noindent\textbf{No regularization.} Our algorithm does not use any form of regularization. Prior work uses either entropic regularization with temperature annealing which introduces a bias (see Thm.~2 in \cite{gemp-et-al25convex-mgs} for cMGs) or quadratic occupancy regularization 
in zero-sum cMGs (see Thm.~4.5 in \cite{kalogiannis-et-al25zs-mgs}). 

\noindent\textbf{Decentralized play.} Algorithm~\ref{algo} can be executed in a decentralized way: each agent~$i$ updates their policy, communicating with other agents~$j\neq i$ to obtain estimates of their marginal occupancies~$\lambda_{\mu,j}^{\pi}$.  When~$F_i$ only depends on $\lambda_{\mu,i}^{\pi}$ (and not on the argument~$\pi_{-i}$), the gradient of agent $i$ simplifies and becomes independent from other agents' occupancies. Agent~$i$ only needs to estimate their own marginal occupancy and updates their policy \textit{independently}. 

\noindent\textbf{Simultaneous play.} In a GUMG,  Algorithm~\ref{algo} can be run simultaneously by $N$ players who do not need to take turns. In contrast, \cite{kalogiannis-et-al25zs-mgs} propose nested and alternating algorithms for 2-player zero-sum cMGs where players need to take turns, i.e. one player updates their policy while the other waits until the other finishes before updating their own policy, in an asynchronous manner.

\section{LEARNING NASH EQUILIBRIA IN POTENTIAL GUMGs} 
\label{sec:potential-convex-MGs}

In this section, we show how to learn approximate NE efficiently in GUMGs using Algorithm~\ref{algo} in the fully cooperative setting where the utilities of the agents enjoy a potential structure. We establish iteration and sample complexity results under different model access assumptions. 
We start by defining potential GUMGs which generalize Markov potential games. 

\begin{definition}[Potential GUMG]
\label{def:cMPG}
A GUMG as defined in Definition~\ref{def:gumg} is said to be potential if there exists a function~$\Phi: \Lambda \to \R$ s.t. for every~$i \in \mathcal{N}$, for every policy~$\pi_i' \in \Pi_i$, joint policy $\pi = (\pi_i, \pi_{-i}) \in \Pi,$
$u_{\mu,i}(\pi_i', \pi_{-i}) - u_{\mu,i}(\pi_i, \pi_{-i})
= \Phi(\lambda_{\mu,1:N}^{\pi_i', \pi_{-i}}) - \Phi(\lambda_{\mu,1:N}^{\pi_i, \pi_{-i}})\,,$
where $\lambda_{\mu,1:N}^{\pi_i', \pi_{-i}} = (\lambda_{\mu,1}^{\pi_i', \pi_{-i}}, \cdots, \lambda_{\mu,N}^{\pi_i', \pi_{-i}})\,.$
\end{definition}

Standard Markov potential games (MPGs) are particular cases of GUMGs, including non-identical interest examples. Beyond MPGs, even the common interest setting for GUMGs (where all the functions~$F_i$ are identical) captures interesting examples that are not instances of MPGs. These include all convex functions of occupancy measures (e.g. KL divergence in imitation learning, entropy in pure exploration) that are not captured by standard MPGs. See section~\ref{subsec:examples} for examples and appendix~\ref{app:limitations-potential-gumgs} for limitations. 

Using the potential function~$\Phi$ defined in \eqref{def:cMPG}, we introduce the function~$\phi_{\mu}: \Pi \to \mathbb{R}$ defined for every joint policy $\pi \in \Pi$ as follows: $\phi_{\mu}(\pi) := \Phi(\lambda_{\mu,1}^{\pi}, \cdots, \lambda_{\mu,N}^{\pi})\,.$ 

 For a potential GUMG, it is worth noting that for any joint policy~$\pi$ and any $i \in \mathcal{N}$,  
 $v_i(\pi) = \nabla_{\pi_i} u_{\mu,i}(\pi) = \nabla_{\pi_i} \phi_{\mu}(\pi)\,,$  
 i.e. the pseudo-gradient field~$v$ defined in \eqref{eq:pseudo-grad} is a gradient field ($v = \nabla \phi_{\mu}$).  
 Using this key observation, we establish the iteration complexity of Algorithm~\ref{algo}. 
 We make the following smoothness assumption on the player utilities~$F_i$ which is a multi-player version of the standard smoothness assumption made in the RL with general utility literature \citep{zhang-et-al21rlgu,barakat-et-al23rl-gen-ut,barakat-et-al25rlgu}.  

\begin{assumption}[Smoothness]
\label{as:smoothness}
For every~$i \in \mathcal{N}$, the function~$F_i$ satisfies two requirements. First, there exists~$l_{\infty} > 0$ s.t. for all~$j \in [N]$, $\lambda_{1:N} = (\lambda_1, \cdots, \lambda_N) \in \Lambda, \|\nabla_{\lambda_j} F_i(\lambda_{1:N}, \pi_{-i})]\|_{\infty} \leq l_{\infty}\,.$ Second, there exists~$L > 0$ s.t. for all $\lambda, \lambda' \in \Lambda$, and for all $\pi_{-i}, \pi_{-i}' \in \Pi_{-i},$
$\|\nabla_{\lambda_j} F_i(\lambda, \pi_{-i}) - \nabla_{\lambda_j} F_i(\lambda', \pi_{-i}')\|_{\infty}\\
\leq L (\|\lambda - \lambda'\|_1 + \|\pi_{-i} - \pi_{-i}'\|_1)\,.$
\end{assumption}

\begin{remark}
The first requirement in Assumption~\ref{as:smoothness} follows from compactness of the space of occupancy measures and continuity of all the partial derivatives of the functions~$F_i\,.$ We mainly state it formally to introduce the constant~$l_{\infty}.$ 
\end{remark}

Under Assumption~\ref{as:smoothness}, we obtain smoothness of the utility functions~$u_{\mu,i}$ for $i \in \mathcal{N}$. 
\begin{proposition}
\label{prop:smoothness}
Let Assumption~\ref{as:smoothness} hold.
Then, for all $i \in \mathcal{N},$ the utility $u_{\mu,i}$ defined in \eqref{eq:def-u_i} is smooth, i.e. there exists a constant $\beta > 0$ s.t. for all $\pi, \pi' \in \Pi,$ 
$ \|\nabla_{\pi} u_{\mu,i}(\pi) - \nabla_{\pi} u_{\mu,i}(\pi')\|_2 \leq \beta \|\pi - \pi'\|_2\,.$   
\end{proposition}

\subsection{Exact Setting Iteration Complexity}

In this section, we suppose that each agent~$i \in \mathcal{N}$ has access to exact gradients~$\nabla_{\pi_i} u_{\mu,i}(\pi)$ for any joint policy~$\pi$ and we analyze the iteration complexity of Algorithm~\ref{algo} to obtain an $\varepsilon$-approximate NE. 

Define for any joint policy~$\pi \in \Pi$ and any $i \in \mathcal{N}$ the individual and maximal NE gaps: 
\begin{align}
\text{NE-Gap}_i(\pi) &:= \underset{\pi_i' \in \Pi_i}{\max} \{ u_{\mu,i}(\pi_i', \pi_{-i}) - u_{\mu,i}(\pi)\}\,,\\ 
\text{NE-Gap}(\pi) &:= \underset{i \in \mathcal{N}}{\max} \,\text{NE-Gap}_i(\pi)\,.
\end{align}
A policy~$\pi^{\star}$ is an $\epsilon$-NE if and only if $\text{NE-Gap}(\pi^{\star}) \leq \epsilon\,.$

\begin{theorem}[Deterministic setting]
\label{thm:iter-complexity-exact-PG} Let Assumptions~\ref{as:concavity}, \ref{as:exploration}, and \ref{as:smoothness} hold. Then, the sequence $(\pi^t)$ of policies generated by running Algorithm~\ref{algo} (with exact gradients in step 8 and $\alpha = 0$) simultaneously by all agents $i \in \mathcal{N}$ with exact gradients and stepsize $\eta \leq 1/\beta$ in a potential GUMG satisfies: 
$\lim_{t \to +\infty} \text{NE-Gap}(\pi^t) = 0\,.$ 
Moreover, 
for any $\varepsilon > 0$, for $T \geq \frac{32 \mathcal{C}_{\mathcal{G}}^2 \beta (\Phi_{\max} - \Phi_{\min})}{(1-\gamma)^2\varepsilon^2}, \frac{1}{T} \sum_{t=1}^T \text{NE-Gap}(\pi^t) \leq \varepsilon\,.$
\end{theorem}

\begin{remark} 
Compared to using a centralized single-agent PG algorithm for convex RL, the iteration complexity only scales with $\sum_{i=1}^N |\cA_i|$ (which appears in $\beta$, see appendix) rather than the cardinality of the joint action space $|\cA|$ which scales exponentially with the number~$N$ of agents. This additive dependence breaks the so-called curse of multi-agency thanks to decentralization. 
As for the convergence rate, notice that our gradient domination (Prop.~\ref{prop:agentwise-grad-dom}) only holds per-agent and is therefore weaker than its single-agent counterpart. Therefore, it cannot be immediately leveraged to obtain faster last-iterate convergence rates like in the single-agent setting \citep{xiao-et-al22pg}. 
\end{remark}

\subsection{Generative Model Sample Complexity}
\label{sec:generative model}

We now suppose that we only have access to a generative model, i.e. a simulator which we can query and samples a next state~$s' \sim P(\cdot|s,a)$ given \textit{any} current state~$s$ and any joint action~$ a \in \cA\,.$ Therefore finite-length state-action trajectories can be sampled from \textit{any} state. Using this model and given our PG theorem (Thm.~\ref{prop:individual-pg}-(ii)) we estimate individual utility gradients as follows: Given a current policy~$\pi$, generate for each~$i \in \mathcal{N}, (s,a_i) \in \cS \times \cA_i$ a set of $M$ independent truncated trajectories of length~$H$: $\tau_{s,a_i}^{(k)} = \{ (s_0^{(k)}, a_{i,0}^{(k)}), \cdots, (s_{H-1}^{(k)}, a_{i,H-1}^{(k)})\}, k \in [M]$ with $s_0^{(k)} = s, a_{i,0}^{(k)} = a_i$. Then compute the  estimates: 
\begin{equation}
\label{eq:stochastic-pg-gen-model}
[\hat{\nabla}_{\pi_i} u_{\mu,i}(\pi)]_{s,a_i} = \frac{1}{1-\gamma} \hat{d}_{\mu}^{\pi}(s) \sum_{j=1}^N\hat{q}^{\pi}_{i,j, (s,a_i)}\,, 
\end{equation}
where $\hat{q}^{\pi}_{i,j, (s,a_i)} = \frac{1}{M} \sum_{k=1}^M \sum_{t=0}^{H-1} \gamma^t \hat{r}_{i,j}^{\pi}(s_t^{(k)},a_{t,j}^{(k)})\,,$ 
and $\hat{r}_{i,j}^{\pi} = \nabla_{\lambda_j}F_i(\hat{\lambda}_{\mu,1:N}^{\pi})$ and $\hat{d}_{\mu}^{\pi}(s),\hat{\lambda}_{\mu,i}^{\pi}(s,a_i)$ are estimated as in \eqref{eq:occup-estimate}.  
Using the above estimated gradients, we are ready to state our first sample complexity result to find an $\varepsilon$-approximate NE in a potential GUMG.  

\begin{theorem}
\label{thm:sample-complexity-gen-model}
Let Assumptions~\ref{as:concavity}, \ref{as:exploration}, and \ref{as:smoothness} hold. 
For any desired accuracy $\varepsilon > 0$ and any confidence level~$\delta \in (0,1)$, if each agent $i \in \mathcal{N}$ runs Algorithm~\ref{algo} (using \eqref{eq:stochastic-pg-gen-model} in step 8) with stepsize~$\eta$, greediness parameter~$\alpha = 0$, minibatch size~$M = \tilde{\mathcal{{O}}}(\varepsilon^{-2})$ and truncation horizon~$H = \mathcal{O}((1-\gamma)^{-1} \log(1/\varepsilon))$ for a number of iterations~$T = \tilde{\mathcal{{O}}}(\varepsilon^{-2})$, then we have $\frac{1}{T} \sum_{t=0}^{T-1} \text{NE-Gap}(\pi^{t+1}) \leq \varepsilon$ with probability at least~$1-\delta\,.$ In particular, the total sample complexity in terms of state-action pairs is $T \times M \times H = \tilde{\mathcal{O}}(\varepsilon^{-4})\,.$\footnote{The notation~$\tilde{\mathcal{O}}(\cdot)$ hides polynomial factors in $(1-\gamma)^{-1}, |\cS|, \max_{i \in \mathcal{N}}|\cA_i|$ and polylogarithmic factors in $1/\varepsilon$ and $|\cS|\cdot \max_{i \in \mathcal{N}} |\cA_i|/\delta$, see appendix for explicit constants.} 
\end{theorem}

\subsection{On-policy Learning Sample Complexity} 
\label{sec:on-policy-samp-comp}

In this last section, we relax the sampling model of section~\ref{sec:generative model} and we suppose that each agent can only sample trajectories of finite length using their current policy at each time step~$t$ as described in section~\ref{sec:pg-algo}. In particular, agents cannot query the transition model from \textit{any} state, including non-visited ones. This justifies the use of $\alpha$-greedy policies with $\alpha > 0$ to guarantee state exploration. 

In contrast to section~\ref{sec:generative model}, the algorithm described in this section does not require a generative model. In section~\ref{sec:generative model}, policy gradients are estimated (using the gradient representation in Prop.~\ref{prop:individual-pg}-\ref{prop-ii:individual-pg}) by computing estimates of the Q functions and occupancy measures \textit{at every state-action pair} $(s,a_i)$ for agent $i$. This requires being able to sample the next state from any given state–action pair, hence the need for a generative model (see e.g. section 5.1 in \cite{xiao-et-al22pg} for a similar model in the single-agent setting). In this section, we only need to estimate the pseudo-rewards at the state-action pairs that were encountered in the sampled trajectories (see Prop.~\ref{prop:individual-pg}-\ref{prop-i:individual-pg}) since only the values of $\bar{r}_{i,j}^{\pi}(s_t,a_{j,t})$ are required to compute stochastic policy gradients. Computing estimates of the occupancy measures at every state-action pair is not needed. The sampling model in section~\ref{sec:on-policy-samp-comp} only allows to sample trajectories using the current policy from initial states drawn from the fixed initial state distribution. 

The next result shows that we can still learn approximate NE in this case while the sample complexity slightly deteriorates compared to Thm.~\ref{thm:sample-complexity-gen-model}.

\begin{theorem}
\label{thm:sample-complexity-on-policy}
Let Assumptions~\ref{as:concavity}, \ref{as:exploration}, and \ref{as:smoothness} hold. 
For any desired accuracy $\varepsilon > 0$, if each agent $i \in \mathcal{N}$ runs Algorithm~\ref{algo} (using \eqref{eq:stochastic-pg} in step 8) in a potential GUMG with stepsize~$\eta$, positive greediness policy parameter~$\alpha = \mathcal{O}(\varepsilon)$, minibatch size~$M = \tilde{\mathcal{{O}}}(\varepsilon^{-3})$ and truncation horizon~$H = \mathcal{O}((1-\gamma)^{-1} \log(1/\varepsilon))$ for a number of iterations~$T = \tilde{\mathcal{{O}}}(\varepsilon^{-2})$, then we have $\E[\text{NE-Gap}(\pi^{\tau})] \leq \varepsilon$ where expectation is w.r.t. both randomness in Algorithm~\ref{algo} due to trajectory sampling and $\tau$ which is a uniform random variable over~$[T]\,.$ In particular, the total sample complexity in terms of state-action pairs is $T \times M \times H = \tilde{\mathcal{O}}(\varepsilon^{-5})$\,.
\end{theorem}

\section{CONCLUSION}
In this work, we established the existence of NE in GUMGs  as fixed points of the projected pseudo-gradient dynamics via Brouwer’s theorem and a per-agent gradient domination property. We designed a model-free policy gradient method and proved the first efficient NE learning guarantees in potential GUMGs. Future directions include extending our analysis to large state–action spaces with function approximation, and exploring generalizations to general-sum settings, e.g. convergence to correlated equilibria under decentralized dynamics \citep{erez-et-al23icml,cai-et-al24aistats}.

\subsubsection*{Acknowledgements}
 
This work is supported by the MOE Tier 2 Grant (MOE-T2EP20223-0018), Ministry of Education Singapore (SRG ESD 2024 174), the CQT++ Core Research Funding Grant (SUTD) (RS-NRCQT-00002), and partially by Project MIS 5154714 of the National Recovery and Resilience Plan, Greece 2.0, funded by the European Union under the NextGenerationEU Program. 
Ioannis Panageas is supported by NSF grant CCF- 2454115.

\bibliographystyle{apalike}
\bibliography{references}

\onecolumn
\aistatstitle{Supplementary Material:\\ 
\medskip 
Convex Markov Games and Beyond: New Proof of Existence, Characterization and Learning Algorithms for Nash Equilibria}

\tableofcontents

\section{EXTENDED RELATED WORK}
\label{sec:extended-related-work}

\noindent\textbf{Markov Games.} Recent years have witnessed active research in theoretical MARL. 
In particular, an important body of work has focused on 
learning approximate Nash equilibria in structured classes of Markov games such as zero-sum Markov games \citep{perolat-et-al15icml,bai-jin20icml,bai-jin-yu20neurips,xie-et-al20colt,daskalakis-foster-golowich20,sayin-et-al21neurips,qiu-et-al21icml,alacaoglu-et-al22icml,cui-du22neurips,zhang-et-al22neurips-pg,zhao-et-al22aistats,park-zhang-ozdaglar23neurips,zhang-et-al23jmlr,cai-et-al23neurips-bandit-zsmg}, Markov potential games \citep{leonardos-et-al22iclr,zhang-et-al24tac,ding-et-al22icml,fox-et-al22aistats,zhang-et-al22neurips,guo-et-al25tac-alpha,sun-et-al23neurips,cheng-et-al24aistats,maheshwari-et-al25tac,alatur-barakat-he24cdc,aydin-eksin23l4dc,sahitaj-et-al25perfMPG} or adversarial team games \citep{kalogiannis-et-al23iclr,kalogiannis-yan-panageas24neurips}. Our work extends these research efforts to cMGs and beyond. 

\noindent\textbf{Convex Markov Games.} Recently, \cite{gemp-et-al25convex-mgs} introduced cMGs, proved existence of NE and proposed to approximate them using a \textit{biased} entropic regularization approach under \textit{full knowledge} of the transition model. In this work, we design a \textit{model-free} PG algorithm for GUMGs and analyze its iteration and \textit{sample complexity}. In a follow-up work, \cite{kalogiannis-et-al25zs-mgs} studied \textit{2-player zero-sum cMGs}. In contrast, we characterize NE in any $N$-player general-sum GUMG (with cMGs as particular cases). Moreover, we study the common-interest setting beyond the competitive zero-sum two-player setting. See sections~\ref{sec:structure-general-cMGs} and~\ref{sec:pg-algo} for more details comparing our theoretical and algorithmic contributions to prior work. Recently, \cite{zamboni-et-al25} studied task-agnostic exploration in MARL, focusing on an identical-interest setting. Their work provides a global optimality iteration complexity for projected gradient ascent leveraging the single-agent analysis of \cite{zhang-et-al21rlgu} under access to exact gradients and assuming smoothness of the objective (see appendix B.1 therein). We note that such a single-agent analysis leads to an exponential dependence on the number of players in general (via the product of action space sizes). In contrast, in our more general setting of GUMGs, we provide a multi-agent decentralized policy gradient sample complexity analysis which only scales with the sum of the action space sizes, using stochastic gradients rather than exact ones and proving smoothness of the objectives under suitable assumptions on the utility functions. 
Note that in general global optimality results cannot be established in our general setting using single-agent analysis \citep{zhang-et-al20variational,zhang-et-al21rlgu} since multi-agent gradient domination (Prop.~\ref{prop:agentwise-grad-dom}) is weaker than the corresponding result for single agents \citep{xiao-et-al22pg} as it only holds per-player.  

\noindent\textbf{Convex MDPs and RL with General Utilities.} 
In the single-agent setting ($N=1$), convex Markov games reduce to convex RL which is an active research topic \citep{hazan-et-al19,zahavy-et-al21,zhang-et-al20variational,zhang-et-al21rlgu,geist-et-al22,agarwal-et-al22,mutti-et-al23convexRL-jmlr,barakat-et-al23rl-gen-ut,marin-et-al24metacurl,desanti-et-al24}. 
A few works consider MARL with general utilities as a \textit{decentralized policy optimization} problem \citep{zhang-et-al22marl-gu-decentralized,ying-et-al23scalable-safe-marl-rlgu}. Our setting is fundamentally different: we rather consider strategic agents in a game-theoretic setting where each agent has their own objective. In particular, GUMGs cannot be cast as single-objective policy optimization or consensus problems in general. Furthermore, our state space is not decomposable as a product of local state spaces. 

\section{MORE DETAILS ON EXAMPLES, EQUILIBRIA AND ALGORITHM} 

\subsection{More Details on Examples of Section~\ref{subsec:examples}}
\label{app:more-about-examples}

In this section we provide more details regarding each one of the examples of section~\ref{subsec:examples} to show how to apply our framework and demonstrate its versatility, especially for gradient computation in Algorithm~\ref{algo} (step 7). 

Recall from~\eqref{eq:ex-composite} our composite example defined for any $\lambda = (\lambda_1, \cdots, \lambda_N) \in \Lambda, \pi = (\pi_i, \pi_{-i}) \in \Pi,$ by: 
\begin{multline}
F_i(\lambda_1, \cdots, \lambda_N, \pi_{-i}) = 
 \underbrace{ - \alpha_i \text{KL}(\lambda_i || q_i)}_{\text{imitation}} \quad \underbrace{- \beta_i \text{KL}(\lambda_i || \bar{\lambda}_N)}_{\text{consensus/diversity}}
+ \underbrace{\gamma_i h_i \left(\sum_{j=1}^N W_{i,j} \lambda_j \right)}_{\text{team aggregate}} + \underbrace{g_i(\pi_{-i})}_{\text{diversity penalty}}\,,
\end{multline}
where $q_i \in \Lambda_i, i \in \mathcal{N}$ are fixed probability distributions, $\alpha_i, \beta_i, \gamma_i$ are nonnegative reals, $W \in \mathbb{R}^{N \times N}$ is a stochastic matrix of weights,
and $\bar{\lambda}_N = \sum_{i=1}^N w_i \lambda_i$ with $w_i \geq 0, \sum_{i=1}^N w_i = 1\,.$ Here, $h_i$ and $g_i$ are concave functions. Note that $F_i$ satisfies Assumption~\ref{as:concavity} since it is the sum of jointly concave functions, noticing that $\text{KL}$ divergence is jointly convex and the composition of a convex function with a linear one is also convex.  

We now elaborate on special cases of our composite example:   

\noindent\textbf{a. Multi-agent imitation learning.} In this case ($\alpha_i > 0,\beta_i = \gamma_i = 0$), the utility function of each agent and its gradient with respect to their own occupancy measure variable are given by: 
\begin{align}
F_i(\lambda_1, \cdots, \lambda_N, \pi_{-i}) &:= - \alpha_i \text{KL}(\lambda_i || q_i) + g_i(\pi_{-i})\,, \nonumber\\
\nabla_{\lambda_i} F_i(\lambda_1, \cdots, \lambda_N, \pi_{-i}) &= - \alpha_i \left( \log\left(\frac{\lambda_i}{q_i}\right) + \mathds{1} \right)\,, \nonumber\\
\nabla_{\lambda_j} F_i(\lambda_1, \cdots, \lambda_N, \pi_{-i}) & = 0\,,\quad \forall j \neq i\,, 
\end{align}
where $\mathds{1}$ refers to the vector of ones and the notation $\log\left(\frac{\lambda_i}{q_i}\right)$ is to be understood component-wise. It remains to use these gradients to compute pseudo-rewards in Algorithm~\ref{algo}. 

\noindent\textbf{b. Consensus and diversity.} In this case ($\beta_i > 0,\alpha_i = \gamma_i = 0$), the utility function of each agent and its gradients with respect to the different occupancy measure variables are given by: 
\begin{align}
F_i(\lambda_1, \cdots, \lambda_N, \pi_{-i}) &:=- \beta_i \text{KL}(\lambda_i || \bar{\lambda}_N) + g_i(\pi_{-i})\,, \nonumber\\
\nabla_{\lambda_i} F_i(\lambda_1, \cdots, \lambda_N, \pi_{-i}) &= - \beta_i \left( \log\left(\frac{\lambda_i}{\bar{\lambda}_N}\right) + \frac{\bar{\lambda}_N - w_i \lambda_i}{\bar{\lambda}_N} \right)\,, \nonumber\\
\nabla_{\lambda_j} F_i(\lambda_1, \cdots, \lambda_N, \pi_{-i}) & = \beta_i \frac{w_j \lambda_i}{\bar{\lambda}_N}\,,\quad \forall j \neq i\,, 
\end{align}
where again $\log$ and fractions are to be computed coordinate-wise. 

\noindent\textbf{c. Population-level coordination.}   
In this case ($\gamma_i > 0,\alpha_i = \beta_i = 0$), the utility function of each agent and its gradients with respect to the different occupancy measure variables are given by: 
\begin{align}
\label{eq:pop-level-coordination}
F_i(\lambda_1, \cdots, \lambda_N, \pi_{-i}) &:=\gamma_i h_i \left(\sum_{j=1}^N W_{i,j} \lambda_j \right) + g_i(\pi_{-i})\,, \nonumber\\
\nabla_{\lambda_i} F_i(\lambda_1, \cdots, \lambda_N, \pi_{-i}) &= \gamma_i W_{i,i}^{\top}  \nabla h_i\left(\sum_{j=1}^N W_{i,j} \lambda_j \right) \,, \nonumber\\
\nabla_{\lambda_j} F_i(\lambda_1, \cdots, \lambda_N, \pi_{-i}) & =  \gamma_i W_{i,j}^{\top}  \nabla h_i\left(\sum_{j=1}^N W_{i,j} \lambda_j \right)\,,\quad \forall j \neq i\,, 
\end{align}
where $\nabla h_i$ is the gradient of the function~$h_i$.

\noindent\textbf{d. Team coverage and distribution matching} 
($F_i(\lambda_{1:N},\pi_{-i}) = - \text{KL}(\bar{\lambda}_N || \lambda_{\text{target}})$). In this case, the utility function of each agent and its gradients with respect to the different occupancy measure variables are given by: 
\begin{align}
F_i(\lambda_1, \cdots, \lambda_N, \pi_{-i}) &:= - \text{KL}(\bar{\lambda}_N || \lambda_{\text{target}})\,, \nonumber\\
\nabla_{\lambda_i} F_i(\lambda_1, \cdots, \lambda_N, \pi_{-i}) &= - w_i \left( \log\left(\frac{\bar{\lambda}_N}{\lambda_{\text{target}}}\right) + \mathds{1} \right) \,, \nonumber\\
\nabla_{\lambda_j} F_i(\lambda_1, \cdots, \lambda_N, \pi_{-i}) & = - w_j \left( \log\left(\frac{\bar{\lambda}_N}{\lambda_{\text{target}}}\right) + \mathds{1} \right) \,,\quad \forall j \neq i\,, 
\end{align}
where again $\log$ and fractions are to be computed coordinate-wise. 

\noindent\textbf{e. Collective exploration.} When $h_i$ is the negative entropy, utility maximization enhances exploration. In this case ($\gamma_i > 0,\alpha_i = \beta_i = 0$), the utility function of each agent and its gradients with respect to the different occupancy measure variables are given by \eqref{eq:pop-level-coordination} and the gradient of the negative entropy is straightforward.   

\subsection{Limitations of Potential GUMGs}
\label{app:limitations-potential-gumgs}

The potential assumption is rigid beyond the identical interest setting. However, we point out that this is already the case for MPGs. Beyond the common interest setting, the potential structure imposes important restrictions on the transitions (e.g. transitions independent of actions). In particular, a state-based potential game (for which reward functions induce a normal-form potential game at each state) is not an MPG in general as further restrictive conditions on state transitions are required (see e.g. Prop. 3.2 in \cite{leonardos-et-al22iclr}). In our case, even the common interest setting is interesting and goes beyond MPGs.

\subsection{Markov Nash Equilibrium vs Markov Perfect Equilibrium}
\label{app:ne-vs-mpe}

A (Markov) Nash equilibrium policy is defined relative to an initial distribution $\mu$ and ensures optimality only for the original experiment starting from that distribution. This policy need not be a Nash policy if the system is perturbed to another state. MPE is a more robust and stronger equilibrium notion than NE.

In many dynamic environments, the system can be unexpectedly pushed into states that were extremely unlikely under the initial distribution, for example, a sudden congestion spike in a network, a robot drifting off a planned trajectory, or a firm experiencing an unexpected shock. A Nash equilibrium defined only for the initial state distribution does not guarantee rational behavior outside the initial state distribution. An MPE is independent from the initial state distribution and guarantees that the policy is a best response in every state, ensuring that the policy remains sensible and stable even when the system experiences disturbances or is initialized outside the initially assumed distribution.

\subsection{Initial State Sampling}

In section~\ref{sec:on-policy-samp-comp}, the initial state in each trajectory is sampled according to the fixed initial state distribution defining the GUMG. This is because the policy gradients are expectations over this initial state distribution (see Prop.~\ref{prop:individual-pg}-\ref{prop-i:individual-pg}) since utilities are defined as functions of the occupancy measures induced by the same initial distribution. Therefore, sampling starting from the initial state distribution yields an unbiased policy gradient estimate. This is consistent with on-policy policy gradient methods for standard MPGs and single-agent RL with general utilities. Whether it is possible to learn Nash policies using a single trajectory by only sampling starting from the current state at each time step is an interesting question that merits further investigation. We expect that this would induce a bias in the policy gradient that needs to be controlled by carefully bounding the policy drift by tuning the step sizes for instance.

\subsection{Utility Function Knowledge}

In most applications of (single-agent) RL with general utilities known in the literature, the function $F_i$ is available, and its gradient can therefore be computed directly. Thus, in the majority of scenarios, access to $F_i$ (and its gradient) is not a restrictive requirement. For instance, this includes negative entropy in pure exploration (see e.g. \cite{hazan-et-al19}), KL divergence in imitation learning (see e.g. \cite{ho-ermon16gail}), A and D-design objectives in experiment design (see e.g. \cite{mutny-et-al23}). In addition, utilities are also typically supposed to be known in recent work on RL with general utilities (e.g. \cite{barakat-et-al23rl-gen-ut,barakat-et-al25rlgu}) and on cMGs \citep{gemp-et-al25convex-mgs,kalogiannis-et-al25zs-mgs}. Investigating settings where it can be useful to use some weaker utility feedback (such as bandit feedback) or learn such utility functions is an interesting direction for future work. 

\subsection{Regularization}

Regularization has been used in recent work on cMGs in \cite{gemp-et-al25convex-mgs} and in the follow-up work of \cite{kalogiannis-et-al25zs-mgs}. \cite{gemp-et-al25convex-mgs} do not adopt our policy gradient viewpoint and assume knowledge of the transition kernel, they use entropic regularization in the utilities and they do not provide convergence guarantees. \cite{kalogiannis-et-al25zs-mgs}  focus on the zero-sum setting and regularization improves  conditioning for solving the resulting minmax problem. In contrast, our work provides convergence guarantees in the potential setting which does not lead to a minmax problem (while our existence and characterization results hold for any GUMG).

Our approach is consistent with existing work in Markov games (Markov potential games in particular) where regularization is not necessary to obtain polynomial sample complexity for learning approximate Nash equilibria, even for 2-player zero-sum Markov games (see e.g. \cite{daskalakis-foster-golowich20}). We believe that we can also prove guarantees for zero-sum GUMGs without regularization using an approach similar to \cite{daskalakis-foster-golowich20} for 2-player zero-sum Markov games. However, such an analysis only provides one-sided Nash guarantees (for one of the players) when running the algorithm and would lead to iteration and sample complexities which would be worse than the analysis of \citet{kalogiannis-et-al25zs-mgs} (which require players to take turns on the other hand). 

Our approach is also inspired from single-agent RL with general utilities policy optimization (e.g. \cite{zhang-et-al21rlgu,barakat-et-al23rl-gen-ut,barakat-et-al25rlgu}) where policy gradient algorithms have been studied without regularization.

\section{PROOFS FOR SECTION~\ref{sec:structure-general-cMGs}}

\subsection{Proof of Proposition~\ref{prop:agentwise-grad-dom}}
\label{sec:proof-prop-agentwise-grad-dom}

\begin{proof}
\textbf{Step 1. Using concavity to relate utility deviation to standard single-agent RL.} 
Let $i \in \mathcal{N}$ and let $\pi = (\pi_i, \pi_{-i}) \in \Pi, \pi_i' \in \Pi_i$. By joint concavity of~$F_i$(Assumption~\ref{as:concavity}) w.r.t.~$(\lambda_1, \cdots, \lambda_N)$, we have that for any $\lambda = (\lambda_1, \cdots, \lambda_N), \lambda' =  (\lambda_1', \cdots, \lambda_N') \in \mathbb{R}^{|\mathcal{S}|\cdot \sum_{j=1}^N |\mathcal{A}_k|}\,,$ and any~$\pi_{-i} \in \Pi_{-i},$
\begin{equation}
\label{eq:joint-concavity-Fi}
F_i(\lambda',\pi_{-i}) - F_i(\lambda,\pi_{-i}) \leq \langle \nabla_{\lambda} F_i(\lambda,\pi_{-i}), \lambda' - \lambda \rangle = \sum_{j=1}^N \langle \nabla_{\lambda_j} F_i(\lambda_1, \cdots, \lambda_N, \pi_{-i}), \lambda_j' - \lambda_j \rangle\,.
\end{equation}

Recall now the definition of the (partial) value function defined for any $j \in \mathcal{N}$, any (joint) policy $\pi \in \Pi$ and any reward vector $r \in \mathbb{R}^{|\mathcal{S}| \cdot |\mathcal{A}_j|}$: 
\begin{equation}
v_{\mu,j}^{\pi}(r) := \langle\lambda_{\mu,j}^{\pi}, r \rangle = \mathbb{E}_{\mu, \pi}\left[ \sum_{t=0}^{\infty} \gamma^t r(s_t, a_{j,t})\right]\,.
\end{equation}
We also define for any $i,j \in \mathcal{N}$ and any joint policy~$\pi \in \Pi$ the pseudo-reward vector: 
\begin{equation}
\bar{r}_{i,j}^{\pi} := \nabla_{\lambda_j} F_i(\lambda_{\mu,1}^{\pi}, \cdots, \lambda_{\mu,N}^{\pi}, \pi_{-i}) \in \mathbb{R}^{|\mathcal{S}|\cdot|\mathcal{A}_j|}\,.
\end{equation}
Using this notation and setting $\lambda'_j = \lambda_{\mu,j}^{\pi_i', \pi_{-i}}$ and  $\lambda_j = \lambda_{\mu,j}^{\pi_i, \pi_{-i}}$ in \eqref{eq:joint-concavity-Fi} yields
\begin{align}
\label{eq:bound-ui-diff}
u_{\mu,i}(\pi_i', \pi_{-i}) - u_{\mu,i}(\pi_i, \pi_{-i}) &= F_i(\lambda_{\mu,1}^{\pi_i', \pi_{-i}}, \cdots, \lambda_{\mu,N}^{\pi_i', \pi_{-i}},\pi_{-i}) - F_i(\lambda_{\mu,1}^{\pi_i, \pi_{-i}}, \cdots, \lambda_{\mu,N}^{\pi_i, \pi_{-i}}, \pi_{-i}) \nonumber\\ 
&\leq \sum_{j=1}^N \langle \bar{r}_{i,j}^{\pi}, \lambda_{\mu,j}^{\pi_i', \pi_{-i}} - \lambda_{\mu,j}^{\pi_i, \pi_{-i}} \rangle \nonumber\\
&= \sum_{j=1}^N v_{\mu,j}^{\pi_i', \pi_{-i}}(\bar{r}_{i,j}^{\pi}) - v_{\mu,j}^{\pi_i, \pi_{-i}}(\bar{r}_{i,j}^{\pi})\,.
\end{align}

\textbf{Step 2. Using single-agent gradient domination in an (averaged) MDP.} 
Given the fixed policy~$\pi_{-i}$, define now for any reward vector $r \in \mathbb{R}^{|\mathcal{S}|\times |\mathcal{A}|}$ the following MDP~$\mathcal{M} := (\mathcal{S}, \mathcal{A}_i, \bar{P}^{\pi_{-i}}, r, \mu, \gamma)$ where the transition kernel is defined as follows: For any $s, s' \in \mathcal{S}, a_i \in \mathcal{A}_i$, 
\begin{equation}
\bar{P}^{\pi_{-i}}(s'|s,a_i) = \sum_{a_{-i} \in \mathcal{A}_{-i}} \pi_{-i}(a_{-i}|s) P(s'|s, a_i, a_{-i})\,, 
\end{equation}
where $P$ is the transition kernel of our initial GUMG. 

In this MDP~$\mathcal{M}$, observe that $v_{\mu,j}^{\pi_i, \pi_{-i}}(r)$ (as defined in \eqref{eq:def-vj-rj}-\eqref{eq:def-tilde-v-tilde-r}) is the value function associated to policy~$\pi_i$ and reward $r\,.$ Using the variational gradient domination property for (single-agent) MDPs (see Proposition~\ref{prop:var-grad-dom-mdps}), we obtain that for every~$\pi_i' \in \Pi_i$ and any state distribution~$\rho \in \Delta(\mathcal{S}),$ 
\begin{equation}
\label{eq:grad-dom-agentwise-proof}
v_{\mu,j}^{\pi_i^{\star},\pi_{-i}}(r) - v_{\mu,j}^{\pi_i,\pi_{-i}}(r)
\leq \left\| \frac{d_{\mu}^{\pi_i^{\star}, \pi_{-i}}}{d_{\rho}^{\pi_i, \pi_{-i}}} \right\|_{\infty}\,   \langle \nabla_{\pi_i} v_{\rho,j}^{\pi}(r) , \pi_i' - \pi_i \rangle\,,
\end{equation}
where~$\pi_i^{\star}$ is an optimal policy in the MDP~$\mathcal{M}$ (depending on the fixed policy~$\pi_{-i}$)\,. 
As a consequence, using \eqref{eq:bound-ui-diff} and \eqref{eq:grad-dom-agentwise-proof} with $r = \bar{r}_{i,j}^{\pi}$, it follows from summing up the above inequality for $j \in \mathcal{N}$ that 
\begin{equation}
u_{\mu,i}(\pi_i', \pi_{-i}) - u_{\mu,i}(\pi_i, \pi_{-i}) 
\leq C_{\mathcal{G}}  \left\langle \sum_{j=1}^N \nabla_{\pi_i} v_{\rho,j}^{\pi}(\bar{r}_{i,j}^{\pi}) , \pi_i' - \pi_i \right\rangle = C_{\mathcal{G}} \left\langle \nabla_{\pi_i} u_{\rho,i}(\pi) , \pi_i' - \pi_i \right\rangle \,,
\end{equation}
where the last identity follows from using the policy gradient expression proved in Proposition~\ref{prop:individual-pg} (see \eqref{eq:pg-sum-grad}). To conclude, it remains to take the maximum over~$\pi_i'$ on the right-hand side of the above inequality.  
\end{proof}

\subsection{Proof of Proposition~\ref{prop:fos-nash-equivalence}}

\textbf{(i) If $\pi^{\star}$ is a Nash policy then $\pi^{\star}$ is FOS.}
\begin{proof}
This implication is standard in optimization. We provide a proof here for completeness. Let~$\pi^{\star}$ be a Nash policy. By definition of the differentiability of~$u_i$ w.r.t~$\pi_i$ at~$\pi^{\star}$, we can write for any $\eta > 0$ and any policy $\pi_i' \in \Pi_i$ (using $\delta_i := \pi_i' - \pi_i^{\star}$) that 
\begin{equation}
u_i(\pi_i^{\star} + \eta \delta_i, \pi_{-i}^{\star}) = u_i(\pi_i^{\star}, \pi_{-i}^{\star}) + \eta\, \langle \nabla_{\pi_i} u_i(\pi^{\star}), \delta_i \rangle + \eta \|\delta_i\| \cdot o_{\eta \to 0}(1)\,.
\end{equation}
Therefore, we obtain by rearranging the identity
\begin{equation}
\langle \nabla_{\pi_i} u_i(\pi^{\star}), \pi_i' - \pi_i^{\star} \rangle = \frac 1\eta (u_i(\pi_i^{\star} + \eta \,\delta_i, \pi_{-i}^{\star}) - u_i(\pi_i^{\star},\pi_{-i}^{\star})) - \|\delta_i\| \cdot o_{\eta \to 0}(1)\,.
\end{equation}
By the definition of a Nash equilibrium policy, it follows that 
\begin{equation}
\langle \nabla_{\pi_i} u_i(\pi^{\star}), \pi_i' - \pi_i^{\star} \rangle \leq - \|\delta_i\| \cdot o_{\eta \to 0}(1)\,.
\end{equation}
Taking $\eta \to 0$ yields $\langle \nabla_{\pi_i} u_i(\pi^{\star}), \pi_i' - \pi_i^{\star} \rangle \leq 0$, i.e.$\langle v_i(\pi^{\star}), \pi_i' - \pi_i^{\star} \rangle \leq 0\,.$ As $i$ is arbitrary in $\mathcal{N}$, the inequality holds for all $i \in \mathcal{N}$ and summing up these inequalities yields $\langle v(\pi^{\star}), \pi' - \pi^{\star} \rangle \leq 0$ for all $\pi' \in \Pi$ which means that $\pi^{\star}$ is a FOS policy. 
\end{proof}

\textbf{(ii) If $\pi^{\star}$ is a FOS policy then  $\pi^{\star}$ is also a Nash policy.}
\begin{proof}
The result follows from our gradient domination result of Proposition~\ref{prop:agentwise-grad-dom}. 

Let~$\pi^{\star} = (\pi_i^{\star}, \pi_{-i}^{\star}) \in \Pi$ be a FOS policy. By definition it satisfies for all $\pi \in \Pi, \langle v(\pi^{\star}), \pi - \pi^{\star} \rangle \leq 0$ for all $\pi \in \Pi\,.$ Choose $\pi = (\pi_i', \pi_{-i}^{\star})$ in this FOS inequality to obtain $\langle v_i(\pi^{\star}), \pi_i' - \pi_i^{\star} \rangle \leq 0$ for all $i \in \mathcal{N}$ and for all $\pi_i' \in \Pi_i\,.$ 
By Proposition~\ref{prop:agentwise-grad-dom}, for all $i \in \mathcal{N}$ and for all $\pi_i' \in \Pi_i,$ we have 
\begin{equation}
\label{eq:grad-dom-proof-eq}
u_i(\pi_i', \pi_{-i}^{\star}) - u_i(\pi_i^{\star}, \pi_{-i}^{\star}) \leq C_{\mathcal{G}} \cdot \max_{\tilde{\pi}_i \in \Pi_i} \langle v_i(\pi^{\star}), \tilde{\pi}_i - \pi_i^{\star} \rangle \,.
\end{equation}
This in turn implies that $u_i(\pi_i', \pi_{-i}^{\star}) - u_i(\pi_i^{\star}, \pi_{-i}^{\star}) \leq 0$ for all $i \in \mathcal{N}$ and for all $\pi_i' \in \Pi_i$ and hence $\pi^{\star}$ is a Nash policy. 
\end{proof}

\subsection{Proof of Lemma~\ref{prop:characterization-fos-fixed-point}}
\label{app:proof-lemma-fos-fixed-pts}

\begin{lemma}
\label{prop:characterization-fos-fixed-point}
A policy~$\pi^{\star}$ is \ref{eq:def-fos} if and only if it is a fixed point of the projected gradient dynamics, i.e. it satisfies: 
$\pi^{\star} = \text{Proj}_{\Pi}(\pi^{\star} + \eta\, v(\pi^{\star})),\forall \eta > 0\,.$
\end{lemma}

\begin{proof}
Suppose that $\pi^{\star} = \text{Proj}_{\Pi}(\pi^{\star} + \eta\, v(\pi^{\star})),\forall \eta > 0\,.$ By characterization of the projection, we have 
\begin{equation}
\label{eq:proj-charac}
\forall \pi \in \Pi, \forall \eta > 0, \langle \pi - \pi^{\star} , (\pi^{\star} + \eta\, v(\pi^{\star})) - \pi^{\star} \rangle \leq 0\,. 
\end{equation}
By simplifying and dividing by $\eta >0$ the above inequality, 
we obtain the desired variational inequality,  
\begin{equation}
\forall \pi \in \Pi, \forall \eta > 0, \langle \pi - \pi^{\star} ,   v(\pi^{\star}) \rangle \leq 0\,,   
\end{equation}
i.e. $\pi^{\star}$ is a FOS policy. The converse implication also holds since \eqref{eq:proj-charac} is a characterization of the projection, i.e. if $\pi^{\star} \in \Pi$ satisfies \eqref{eq:proj-charac} then it must be the projection of $\pi^{\star} + \eta\, v(\pi^{\star})$ on the policy space~$\Pi$. 
\end{proof}

\subsection{Proof of Theorem~\ref{prop:equiv-ne-fixed-point-projgrad}}

The proof follows from combining Prop.~\ref{prop:fos-nash-equivalence} and Lem.~\ref{prop:characterization-fos-fixed-point}. Prop.~\ref{prop:fos-nash-equivalence} shows that a (joint) policy $\pi^{\star} \in \Pi$ is a Nash policy in a GUMG if and only if it is FOS. Lem.~\ref{prop:characterization-fos-fixed-point} shows that a policy is FOS if and only if it is a fixed point of the projected gradient dynamics. 

\subsection{Proof of Corollary~\ref{cor:existence-NE-cMG}} 

The proof of existence of a NE is in the main part. We now use this result to further show the existence of Markov perfect equilibrium (MPE). 

Let $\varepsilon > 0\,.$ Let $ s \in \cS$ be an arbitrary state. Consider the mixed initial state distribution~$\mu_{\varepsilon} := (1-\varepsilon) \delta_{s} +  \varepsilon \frac{1}{|\cS|}$ where~$\delta_{s}$ is the Dirac measure at the state~$s \in \cS$, defined s.t. $\delta_s(s) = 1$ and $\delta_s(s') = 0$ for any $s' \neq s\,.$ Since $\mu_{\varepsilon}$ is fully mixed and hence implies that Assumption~\ref{as:exploration} is satisfied (i.e. $d_{\mu_{\varepsilon}}^{\pi}(s) > 0$ for all $\tilde{s} \in \cS, \pi \in \Pi$ since it always holds that $d_{\mu_{\varepsilon}}^{\pi}(\tilde{s}) \geq \mu_{\varepsilon}(\tilde{s})$) We can apply the NE existence result using the initial state distribution~$\mu_{\varepsilon}$ established in the first part of Corollary~\ref{cor:existence-NE-cMG} to obtain the existence of a NE~$\pi^{\star}_{\varepsilon}$ (a priori depending on~$\varepsilon$). By definition this NE satisfies the following:
\begin{equation}
\label{eq:ne-ineq}
\forall \pi_i \in \Pi_i,\quad u_{\mu,i}(\pi_i, \pi_{\varepsilon,-i}^{\star}) \leq u_{\mu,i}(\pi^{\star}_{\varepsilon})\,.
\end{equation}
Let $(\varepsilon_n)$ be a sequence of positive numbers converging to zero. The sequence of policies $(\pi^{\star}_{\varepsilon_n})$ evolves in the compact space of policies~$\Pi.$ By the Bolzano-Weierstrass theorem, we can extract from it a converging subsequence~$(\pi^{\star}_{\varepsilon_{\varphi(n)}})$ where $\varphi$ is an increasing function over integers. Denote by $\pi^{\star}$ the limit of the subsequence~$(\pi^{\star}_{\varepsilon_{\varphi(n)}})$. It follows from~\eqref{eq:ne-ineq} that for any integer~$n$, 
\begin{equation}
\label{eq:ne-ineq-extracted}
\forall \pi_i \in \Pi_i,\quad u_{\mu_{\varepsilon_{\varphi(n)}},i}(\pi_i, \pi_{\varepsilon_{\varphi(n)},-i}^{\star}) \leq u_{\mu_{\varepsilon_{\varphi(n)}},i}(\pi^{\star}_{\varepsilon_{\varphi(n)}})\,.
\end{equation}
We now note that the utility function~$u_{\mu,i}$ is continuous w.r.t. both its policy argument and the initial distribution~$\mu$. Indeed, this follows from observing its definition~\eqref{eq:def-u_i}, noting that~$F_i$ is supposed to be continuous and the occupancy measures are continuous functions in both their policy variable (see e.g., Prop.~\ref{prop:lip-state-occup-measures}) and the initial state distribution~$\mu$ (by linearity of the occupancy measure w.r.t.~$\mu$).  
Finally, taking the limit when $n \to \infty$ in~\eqref{eq:ne-ineq-extracted} yields: 
\begin{equation}
\forall \pi_i \in \Pi_i,\quad u_{\delta_s,i}(\pi_i, \pi_{-i}^{\star}) \leq u_{\delta_s,i}(\pi^{\star})\,,
\end{equation}
which proves that~$\pi^{\star}$ is a MPE. 

\begin{remark}
A stronger result can be shown: an NE for a given full-support initial state distribution is also an MPE. Indeed, it can be proved that a NE~$\pi$ for a given full-support initial state distribution~$\mu$ is also a NE for any other full-support initial state distribution~$\rho \neq \mu$ using per-player gradient domination (Prop.~\ref{prop:agentwise-grad-dom}) combined with the policy gradient theorem (Prop.~\ref{prop:individual-pg}-\ref{prop-ii:individual-pg}). Therefore a NE is independent of the initial full-support state distribution. Then a similar limit argument as above shows that this same NE policy~$\pi$ is a MPE. 
\end{remark}

\section{PROOFS FOR SECTION~\ref{sec:pg-algo}: Proof of Proposition~\ref{prop:individual-pg}}

\begin{proof}
We prove each result separately. 

\noindent\textbf{Proof of (i).} 
Using the chain rule, we have for every $i \in \mathcal{N},$
\begin{equation}
\label{eq:grad-ui-1st-expression}
\nabla_{\pi_i} u_{\mu,i}(\pi) = \sum_{j=1}^N [\nabla_{\pi_i} \lambda_{\mu,j}^{\pi}]^T\, \nabla_{\lambda_j} F_i(\lambda_{\mu,1}^{\pi}, \cdots, \lambda_{\mu,N}^{\pi}, \pi_{-i})\,.
\end{equation}

Now define for all $j \in \mathcal{N}, r_j \in \mathbb{R}^{|\mathcal{S}| \cdot |\mathcal{A}_j|},$
\begin{equation}
\label{eq:def-vj-rj}
v_{\mu,j}^{\pi}(r_j) = \langle \lambda_{\mu,j}^{\pi}, r_j \rangle = \mathbb{E}_{\mu, \pi} \left[ \sum_{t=0}^{+\infty} \gamma^t r_j(s_t, a_t^j) \right]\,.
\end{equation}
It follows that 
\begin{equation}
\label{eq:grad-vj}
\nabla_{\pi_i} v_{\mu,j}^{\pi}(r_j) = \sum_{s \in \mathcal{S}, a_j \in \mathcal{A}_j} \nabla_{\pi_i} \lambda_{\mu,j}^{\pi}(s,a_j) r_j(s,a_j) = [\nabla_{\pi_i} \lambda_{\mu,j}^{\pi}]^{\top} r_j\,.
\end{equation}
Therefore, we obtain from \eqref{eq:grad-ui-1st-expression} and \eqref{eq:grad-vj} together that
\begin{equation}
\label{eq:pg-sum-grad}
\nabla_{\pi_i} u_{\mu,i}(\pi) = \sum_{j=1}^N \nabla_{\pi_i} v_{\mu,j}^{\pi}(r_j)_{\big| r_j = \nabla_{\lambda_j} F_i(\lambda_{\mu,1}^{\pi}, \cdots, \lambda_{\mu,N}^{\pi}, \pi_{-i})} = \sum_{j=1}^N \nabla_{\pi_i} v_{\mu,j}^{\pi}(\bar{r}_{i,j}^{\pi})\,,
\end{equation}
where we recall that $\bar{r}_{ij}^{\pi} = \nabla_{\lambda_j} F_i(\lambda_{\mu,1}^{\pi}, \cdots, \lambda_{\mu,N}^{\pi}, \pi_{-i})$ for any $j \in \mathcal{N}\,.$ 
Now define the value function 
\begin{equation}
\label{eq:def-tilde-v-tilde-r}
\tilde{v}_{\mu,j}^{\pi}(\tilde{r}_j) = \mathbb{E}_{\mu, \pi} \left[ \sum_{t=0}^{+\infty} \gamma^t \tilde{r}_j(s_t, a_{j,t}, a_{-j,t}) \right]\,, 
\end{equation}
where $\tilde{r}_j(s_t, a_t) = r_j(s_t, a_{j,t})$ for all $a_t = (a_{j,t}, a_{-j,t})\,.$ 
Given the above notation, we can now apply the standard policy gradient theorem for MDPs (single-agent RL, see Proposition~\ref{prop:classical-pg-thm}) to write 
\begin{equation}
\nabla_{\pi} \tilde{v}_{\mu,j}^{\pi}(\tilde{r}_j) = \mathbb{E}_{\mu, \pi} \left[ \sum_{t=0}^{+\infty} \gamma^t \tilde{r}_j(s_t, a_t) \sum_{t'=0}^{t} \nabla_{\pi} \log \pi(a_{t'}|s_{t'})   \right]\,.
\end{equation}
Now using our product policy form $\left(\pi(a|s) = \prod_{i=1}^N \pi_i(a_i|s) \right)$ and the definition of $\tilde{r}_j$ (which implies that for any policy~$\pi \in \Pi, v_{\mu,j}^{\pi}(r_j)=\tilde{v}_{\mu,j}^{\pi}(\tilde{r}_j)$, see \eqref{eq:def-vj-rj} and \eqref{eq:def-tilde-v-tilde-r}), it follows that 
\begin{equation}
\nabla_{\pi_i} v_{\mu,j}^{\pi}(r_j) = \nabla_{\pi_i} \tilde{v}_{\mu,j}^{\pi}(\tilde{r}_j) =  \mathbb{E}_{\mu, \pi} \left[ \sum_{t=0} ^{+\infty} \gamma^t r_j(s_t, a_{j,t}) \sum_{t'=0}^t \nabla_{\pi_i} \log \pi_i(a_{i,t'}|s_{t'}) \right]\,,
\end{equation}
which gives the desired result when combined with \eqref{eq:pg-sum-grad}.  

\noindent\textbf{Proof of (ii).} 
Notice first from \eqref{eq:pg-sum-grad} that for every state-action pair~$(s,a_i) \in \cS \times \cA_i,$ 
\begin{equation}
\label{eq:chain-rule-sum-pg}
\frac{\partial u_{\mu,i}(\pi)}{\partial \pi_i(a_i|s)} 
= \sum_{j=1}^N  \frac{\partial v_{\mu,j}^{\pi}(r)}{\partial \pi_i(a_i|s)}_{\large{|} r =\bar{r}_{i,j}^{\pi}}\,,
\end{equation}

By the classical policy gradient theorem (see Proposition~\ref{prop:classical-pg-thm}), we have for every~$j \in [N]$ and every $r \in \R^{|S| \cdot |A|},$
\begin{equation}
\nabla_{\pi} v_{\mu,j}^{\pi}(r) = \frac{1}{1-\gamma} \sum_{s' \in \mathcal{S}} \sum_{a' \in \mathcal{A}} \lambda_{\mu}^{\pi}(s',a')\, Q_{s',a'}^{\pi}(r) \, \nabla \log \pi(a'|s')\,.
\end{equation}

Hence, it follows that for every $r \in \R^{|S| \cdot |A|}$ and every state-action pair~$(s,a_i) \in \cS \times \cA_i,$ 
\begin{equation}
\frac{\partial v_{\mu,j}^{\pi}(r)}{\partial \pi_i(a_i|s)} = \frac{1}{1-\gamma} \sum_{s' \in \mathcal{S}} \sum_{a' \in \mathcal{A}} d_{\mu}^{\pi}(s') \pi(a'|s')\, Q_{s',a'}^{\pi}(r) \, \frac{\partial \log \pi(a'|s')}{\partial \pi_i(a_i|s)}\,.
\end{equation}
Since $\pi(a'|s') = \prod_{i=1}^N \pi_i(a_i|s)$ and $\pi_i$ is a direct parameterized policy, we have 
\begin{equation}
\frac{\partial \log \pi(a'|s')}{\partial \pi_i(a_i|s)} = \frac{\partial \log \pi_i(a'|s')}{\partial \pi_i(a_i|s)} = \frac{1}{\pi_i(a_i|s)} \mathbf{1}_{\{s' = s, a_i' = a_i\}}\,,
\end{equation}
where we use the notation~$\mathbf{1}_{\mathcal{X}}$ for the indicator function of the set~$\mathcal{X}\,.$
Therefore, we have: 
\begin{align}
\label{eq:partial-derivative-v_j}
\frac{\partial v_{\mu,j}^{\pi}(r)}{\partial \pi_i(a_i|s)} 
&=  \frac{1}{1-\gamma} \sum_{s' \in \mathcal{S}} \sum_{a' \in \mathcal{A}} d_{\mu}^{\pi}(s') \pi(a'|s')\, Q_{s',a'}^{\pi}(r_i^{\pi}) \frac{1}{\pi_i(a_i|s)} \mathbf{1}_{\{s' = s, a_i' = a_i\}}\, \nonumber\\
&= \frac{1}{1-\gamma} d_{\mu}^{\pi}(s) \sum_{a_{-i}' \in \mathcal{A}_{-i}} \pi_{-i}(a_{-i}'|s) Q_{s, (a_i, a_{-i}')}^{\pi}(r)\nonumber\\
&= \frac{1}{1-\gamma} d_{\mu}^{\pi}(s) \bar{Q}_{s,a_i}^{\pi}(r)\,,
\end{align}
where the average Q-function~$\bar{Q}^{\pi}(r)$ is defined for any joint policy~$\pi \in \Pi$, any reward $r \in \R^{|\cS|\cdot|\cA|}$ and for any state-action pair $(s, a_i) \in \cS \times \cA_i$ as follows:  
\begin{equation}
\label{eq:avg-Q-fun}
\bar{Q}_{s,a_i}^{\pi}(r) := \sum_{a_{-i} \in \cA_{-i}} \pi_{-i}(a_{-i}|s) \,Q_{s,(a_i, a_{-i})}^{\pi}(r)\,.
\end{equation} 

Using \eqref{eq:chain-rule-sum-pg} together with \eqref{eq:partial-derivative-v_j}, we obtain
\begin{equation}
\frac{\partial u_{\mu,i}(\pi)}{\partial \pi_i(a_i|s)} 
= \sum_{j=1}^N  \frac{\partial v_{\mu,j}^{\pi}(r)}{\partial \pi_i(a_i|s)}_{\large{|} r =\bar{r}_{i,j}^{\pi}} 
= \frac{1}{1-\gamma} d_{\mu}^{\pi}(s) \sum_{j=1}^N \bar{Q}_{s,a_i}^{\pi}(\bar{r}_{i,j}^{\pi})\,,
\end{equation}
which concludes the proof. 
\end{proof}

\section{PROOFS FOR SECTION~\ref{sec:potential-convex-MGs}}

\subsection{Proof of Proposition~\ref{prop:smoothness}}

\begin{proof}
Let $i \in \mathcal{N}, \pi, \pi' \in \Pi$\,. We show smoothness by controlling each one of the partial gradients w.r.t. policy~$\pii$ of agent~$i$ since: 
\begin{equation}
\label{eq:sum-of-norms-partials}
\|\nabla_{\pi} u_{\mu,i}(\pi) - \nabla_{\pi} u_{\mu,i}(\pi')\|_2^2 = \sum_{i=1}^N \|\nabla_{\pii} u_{\mu,i}(\pi) - \nabla_{\pii} u_{\mu,i}(\pi')\|_2^2\,.
\end{equation}

Using Prop.~\ref{prop:individual-pg} (see eq.~\eqref{eq:pg-sum-grad}), we can write the following decomposition: 
\begin{align}
\label{eq:decomp-proof-smoothness}
\nabla_{\pi_i} u_{\mu,i}(\pi) - \nabla_{\pi_i} u_{\mu,i}(\pi') 
&= \sum_{j=1}^N \nabla_{\pi_i} v_{\mu,j}^{\pi}(\bar{r}_{i,j}^{\pi})
- \nabla_{\pi_i} v_{\mu,j}^{\pi'}(\bar{r}_{i,j}^{\pi'}) \nonumber\\
&=  \sum_{j=1}^N \nabla_{\pi_i} v_{\mu,j}^{\pi}(\bar{r}_{i,j}^{\pi})
- \nabla_{\pi_i} v_{\mu,j}^{\pi'}(\bar{r}_{i,j}^{\pi}) 
 +  \nabla_{\pi_i} v_{\mu,j}^{\pi'}(\bar{r}_{i,j}^{\pi}) 
 - \nabla_{\pi_i} v_{\mu,j}^{\pi'}(\bar{r}_{i,j}^{\pi'})\,.
\end{align}

Using Lemma~D.7 in \citet{giannou-et-al22neurips} and boundedness of the pseudo-rewards (Assumption~\ref{as:smoothness}), we obtain 
\begin{equation}
\label{eq:term1-smoothness}
\|\nabla_{\pi_i} v_{\mu,j}^{\pi}(\bar{r}_{i,j}^{\pi})
- \nabla_{\pi_i} v_{\mu,j}^{\pi'}(\bar{r}_{i,j}^{\pi})\| 
\leq \frac{3 \sqrt{|\cA_i|} l_{\infty}}{(1-\gamma)^3} \sum_{k=1}^N \sqrt{|\cA_k|} \cdot \|\pi_k - \pi_k'\|_2\,. 
\end{equation}

Using Prop.~\ref{prop:lipschitzness-value-wrt-rewards} and then the smoothness assumption~\ref{as:smoothness} together with the definition of pseudo-rewards, we get 
\begin{equation}
\label{eq:interm-2-proof}
\|\nabla_{\pi_i} v_{\mu,j}^{\pi'}(\bar{r}_{i,j}^{\pi}) 
 - \nabla_{\pi_i} v_{\mu,j}^{\pi'}(\bar{r}_{i,j}^{\pi'})\|_2
 \leq \frac{\sqrt{|\cAi|}}{(1-\gamma)^2} \|\bar{r}_{i,j}^{\pi} - \bar{r}_{i,j}^{\pi'} \|_{\infty}
 \leq \frac{\sqrt{|\cAi|} L}{(1-\gamma)^2} (\|\lambda_{\mu,1:N}^{\pi} - \lambda_{\mu,1:N}^{\pi'}\|_1 + \|\pi_{-i} - \pi_{-i}'\|_1)\,.
\end{equation}

The next lemma shows that policy deviations only scale in norm with the sum of the sizes of the action spaces rather than their product. 
\begin{lemma}
\label{lem:diff-policy-norms}
The following statements hold when using product policies: 
\begin{enumerate}[label=(\roman*)]
\item \label{item1-diff-policy} $\forall s \in \cS,\, \|\pi_{-i}(\cdot|s) - \pi_{-i}'(\cdot|s)\|_1 \leq \sum_{k=1, k \neq i} \|\pi_k(\cdot|s) - \pi_k'(\cdot|s)\|_1\,,$
\item \label{item2-diff-policy} $\forall s \in \cS,\, \|\pi_{-i}(\cdot|s) - \pi_{-i}'(\cdot|s)\|_1 \leq \sqrt{\sum_{k=1}^N |\cA_k|} \cdot \|\pi - \pi'\|_2\,,$

\item \label{item3-diff-policy} $\sum_{k=1}^N  \|\pi_k - \tilde{\pi}_k\|_1 \leq 
\sqrt{|\cS| \cdot \sum_{k=1}^N |\cA_k|} \cdot \|\pi - \tilde{\pi}\|_2\,.$

\item \label{item4-diff-policy} $\|\pi - \tilde{\pi}\|_1 \leq \sqrt{|\cS| \cdot \sum_{k=1}^N |\cA_k|} \cdot \|\pi - \tilde{\pi}\|_2\,.$
\end{enumerate}
\end{lemma}

\begin{proof}
See complete proof in section~\ref{subsec:proof-lem-diff-policy-norms} below. 
\end{proof}

We now upper bound each one of the terms in \eqref{eq:interm-2-proof}. For the first term, we have: 
\begin{align}
\label{eq:diff-occup-meas-bound}
\|\lambda_{\mu,1:N}^{\pi} - \lambda_{\mu,1:N}^{\pi'}\|_1 
&= \sum_{i=1}^N \sum_{s \in \cS, a_i \in \cA_i} |d_{\mu}^{\pi}(s) \pi_i(a_i|s) - d_{\mu}^{\pi'}(s) \pi_i'(a_i|s)| \nonumber\\
&\leq \sum_{i=1}^N \sum_{s \in \cS, a_i \in \cA_i} |d_{\mu}^{\pi}(s) - d_{\mu}^{\pi'}(s)| \pi_i(a_i|s) + d_{\mu}^{\pi'}(s) |\pi_i(a_i|s) - \pi_i'(a_i|s)| \nonumber\\
&\leq N \|d_{\mu}^{\pi} - d_{\mu}^{\pi'}\|_1 + \sum_{i=1}^N \sum_{s \in \cS, a_i \in \cA_i} d_{\mu}^{\pi'}(s) |\pi_i(a_i|s) - \pi_i'(a_i|s)| \nonumber\\
&\leq N \|d_{\mu}^{\pi} - d_{\mu}^{\pi'}\|_1 + \sum_{i=1}^N \|\pi_i - \pi_i'\|_1 \nonumber\\
&\leq \frac{N \gamma}{1-\gamma} \|\pi - \tilde{\pi}\|_1 + \sqrt{|\cS| \cdot \sum_{k=1}^N |\mathcal{A}_k|} \|\pi - \pi'\|_2 \nonumber\\
&\leq \left(1 + \frac{N \gamma}{1-\gamma} \right) \sqrt{|\mathcal{S}| \cdot \sum_{k=1}^N |\mathcal{A}_k|} \|\pi - \pi'\|_2\,, 
\end{align}
where we used Proposition~\ref{prop:lip-state-occup-measures} and Lemma~\ref{lem:diff-policy-norms}. 

Using again Lemma~\ref{lem:diff-policy-norms} yields $\|\pi_{-i} - \pi_{-i}'\|_1 \leq |\mathcal{S}| \cdot \sqrt{\sum_{k=1}^N |\mathcal{A}_k|} \cdot \|\pi - \pi'\|_2$ and we obtain then using \eqref{eq:diff-occup-meas-bound}: 
\begin{equation}
\|\lambda_{\mu,1:N}^{\pi} - \lambda_{\mu,1:N}^{\pi'}\|_1 + \|\pi_{-i} - \pi_{-i}'\|_1 
\leq \sqrt{\sum_{k=1}^N |\mathcal{A}_k|} \left[ \left( 1 + \frac{N \gamma}{1-\gamma}\right) \sqrt{|\cS|} + |\cS| \right] \|\pi - \pi'\|_2\,.
\end{equation} 
Plugging the above inequality into \eqref{eq:interm-2-proof} and then using \eqref{eq:term1-smoothness} and \eqref{eq:decomp-proof-smoothness} yields: 
\begin{multline}
\label{eq:intermed-ineq}
\|\nabla_{\pi_i} u_{\mu,i}(\pi) - \nabla_{\pi_i} u_{\mu,i}(\pi')\|_2 
\leq \frac{3 N \sqrt{|\cA_i|} l_{\infty}}{(1-\gamma)^3} \sum_{k=1}^N \sqrt{|\cA_k|} \cdot \|\pi_k - \pi_k'\|_2\\
+ \frac{N \sqrt{|\cAi|} L}{(1-\gamma)^2} \sqrt{\sum_{k=1}^N |\mathcal{A}_k|} \left[ \left( 1 + \frac{N \gamma}{1-\gamma}\right) \sqrt{|\cS||} + |\cS| \right] \|\pi - \pi'\|_2\,.
\end{multline}
Using Cauchy-Schwarz's inequality gives: 
\begin{equation}
\sum_{k=1}^N \sqrt{|\cA_k|} \cdot \|\pi_k - \pi_k'\|_2  \leq \sqrt{\sum_{k=1}^N |\cA_k|} \cdot \sqrt{\sum_{i=1}^N \|\pi_i - \pi_i'\|_2^2} \leq \sqrt{\sum_{k=1}^N |\cA_k|}  \cdot \|\pi - \pi'\|_2\,. 
\end{equation}
Using the above inequality in \eqref{eq:intermed-ineq}, we obtain: 
\begin{equation}
\|\nabla_{\pi_i} u_{\mu,i}(\pi) - \nabla_{\pi_i} u_{\mu,i}(\pi')\|_2 
\leq \beta_i \|\pi - \pi'\|_2\,,
\end{equation}
where the constant~$\beta_i$ is defined as follows: 
\begin{equation}
\beta_i := \frac{3 N \sqrt{|\cA_i|} l_{\infty}}{(1-\gamma)^3} \sqrt{\sum_{k=1}^N |\cA_k|} 
+ \frac{N \sqrt{|\cAi|} L}{(1-\gamma)^2} \sqrt{\sum_{k=1}^N |\mathcal{A}_k|} \left[ \left( 1 + \frac{N \gamma}{1-\gamma}\right) \sqrt{|\cS|} + |\cS| \right]\,.
\end{equation}

It follows from this that: 
\begin{equation}
\|\nabla_{\pi} u_{\mu,i}(\pi) - \nabla_{\pi} u_{\mu,i}(\pi')\|_2  
\leq \sum_{i=1}^N \|\nabla_{\pi_i} u_{\mu,i}(\pi) - \nabla_{\pi_i} u_{\mu,i}(\pi')\|_2
\leq \left(\sum_{i=1}^N \beta_i \right) \cdot \|\pi - \pi'\|_2\,.
\end{equation}
Using again Cauchy-Schwarz's inequality gives $\sum_{i=1}^N \sqrt{|\cAi|} \leq \sqrt{N \sum_{i=1}^N |\cAi|}$, we have the following bound: 
\begin{align}
\label{eq:smoothness-beta}
\sum_{i=1}^N \beta_i &= \frac{3 N l_{\infty} \sum_{i=1}^N \sqrt{|\cA_i|}}{(1-\gamma)^3} \sqrt{\sum_{k=1}^N |\cA_k|} 
+ \frac{N \sqrt{|\cAi|} L}{(1-\gamma)^2} \sqrt{\sum_{k=1}^N |\mathcal{A}_k|} \left[ \left( 1 + \frac{N \gamma}{1-\gamma}\right) \sqrt{|\cS||} + |\cS| \right]\\
&\leq \beta := \frac{N^{\frac{3}{2}} \sum_{k=1}^N |\mathcal{A}_k|}{(1-\gamma)^2} (3 l_{\infty} + L \left[ \left( 1 + \frac{N \gamma}{1-\gamma}\right) \sqrt{|\cS|} + |\cS| \right])\,.
\end{align}

Finally, we obtain our desired $\beta$-smoothness bound: 
\begin{equation}
\|\nabla_{\pi} u_i(\pi) - \nabla_{\pi} u_i(\pi')\|_2 
\leq \beta  \|\pi - \pi'\|_2\,,
\end{equation}
where the smoothness constant $\beta$ is defined as follows: 
\begin{equation}
\label{eq:def-beta}
\beta := \frac{N^{\frac{3}{2}} \sum_{k=1}^N |\mathcal{A}_k|}{(1-\gamma)^2} \left(3 l_{\infty} + L \left[ \left( 1 + \frac{N \gamma}{1-\gamma}\right) \sqrt{|\cS|} + |\cS| \right] \right)\,,
\end{equation}
and this concludes the proof. 
\end{proof}

\subsection{Proof of Lemma~\ref{lem:diff-policy-norms}}
\label{subsec:proof-lem-diff-policy-norms}

\begin{proof}
We prove each one of the points in what follows.

\noindent\textbf{Item~\ref{item1-diff-policy}.} 
First we use the product form of the joint policy to write: 
\begin{align}
\label{eq:difference-policies1}
|\pi_{-i}(a_{-i}|s) - \pi_{-i}'(a_{-i}|s)| &= \left|  \prod_{j=1, j \neq i}^{N} \pi_j(a_j|s) -  \prod_{j=1, j \neq i}^{N} \pi_j'(a_j|s)\right| \nonumber\\
&= \left| \sum_{k=1, k \neq i}^N (\pi_k - \pi_k')(a_k|s) \prod_{l=1, l \neq i}^{k-1} \pi_l(a_l|s) \prod_{l=k+1, l \neq i}^{N} \pi_l'(a_l|s) \right|\,.
\end{align}
As a consequence, we obtain by summing up the previous inequality over $a_{-i} \in \cA_{-i}$: 
\begin{align}
\label{eq:term1-interm}
\sum_{a_{-i} \in \cA_{-i}} |\pi_{-i}(a_{-i}|s) - \pi_{-i}'(a_{-i}|s)| 
&\leq \sum_{k=1, k \neq i}^N \sum_{a_k \in \cA_k} |(\pi_k - \pi_k')(a_k|s)| \sum_{a_{-k} \in \cA_{-k}} \prod_{l=1, l \neq i}^{k-1} \pi_l(a_l|s) \prod_{l=k+1, l \neq i}^{N} \pi_l'(a_l|s) \nonumber\\
&\leq \sum_{k=1, k \neq i} \|\pi_k(\cdot|s) - \pi_k'(\cdot|s)\|_1 \nonumber\\
&\leq \sum_{k=1, k \neq i} \sqrt{|\cA_k|} \|\pi_k(\cdot|s) - \pi_k'(\cdot|s)\|_2\,, 
\end{align}
where the second inequality follows from noting that $\sum_{a_{-k} \in \cA_{-k}} \prod_{l=1, l \neq i}^{k-1} \pi_l(a_l|s) \prod_{l=k+1, l \neq i}^{N} \pi_l'(a_l|s) = 1$ by the definition of policies.

\noindent\textbf{Item~\ref{item2-diff-policy}.}
Using the first item, we can write for any~$s \in \cS$, 
\begin{align}
\|\pi_{-i}(\cdot|s) - \pi_{-i}'(\cdot|s)\|_1
&\leq \sum_{k=1, k \neq i} \|\pi_k(\cdot|s) - \pi_k'(\cdot|s)\|_1 \nonumber\\
&\leq \sum_{k=1, k \neq i} \sqrt{|\cA_k|} \|\pi_k(\cdot|s) - \pi_k'(\cdot|s)\|_2\,,\nonumber\\ 
&\leq \sqrt{\sum_{k=1}^N |\cA_k|} \cdot \|\pi - \pi'\|_2\,,
\end{align}
where the last inequality follows from the Cauchy-Schwarz inequality. 

\noindent\textbf{Item~\ref{item3-diff-policy}.} 
We upper bound the norm as follows: 
\begin{align}
\sum_{k=1}^N  \|\pi_k - \tilde{\pi}_k\|_1
&\leq \sum_{k=1}^N \sqrt{|\cS| \cdot |\cA_k|} \cdot \|\pi_k - \tilde{\pi}_k\|_2\nonumber\\
&\leq \sqrt{|\cS| \cdot \sum_{k=1}^N |\cA_k|} \sqrt{\sum_{k=1}^N \|\pi_k - \tilde{\pi}_k\|_2^2}\nonumber\\
&= \sqrt{|\cS| \cdot \sum_{k=1}^N |\cA_k|} \cdot \|\pi - \tilde{\pi}\|_2\,,
\end{align}
where the first inequality follows similarly to the derivations in \eqref{eq:difference-policies1} and \eqref{eq:term1-interm}. 

\noindent\textbf{Item~\ref{item4-diff-policy}.} 
Observe using similar derivations as in the proof of the first two items that:  
\begin{equation}
\|\pi - \tilde{\pi}\|_1 = \sum_{s \in \cS, a \in \cA} |\pi(a|s) - \tilde{\pi}(a|s)| 
\leq \sum_{s \in \cS} \sum_{k=1}^N \|\pi_k(\cdot|s) - \tilde{\pi}_k(\cdot|s)\|_1 
= \sum_{k=1}^N  \|\pi_k - \tilde{\pi}_k\|_1\,,
\end{equation}
and use the previous item. 
\end{proof}

\subsection{Proof of Thm.~\ref{thm:iter-complexity-exact-PG}}

\noindent\textbf{Proof sketch.} The proof follows a similar strategy to prior work (e.g. \cite{leonardos-et-al22iclr,zhang-et-al24tac}) and consists of three main steps: 
\begin{enumerate}
\item Relate the Nash gap to first-order stationarity using gradient domination (Proposition~\ref{prop:agentwise-grad-dom}): see end of the proof in this section.  
\item Relate the first-order stationarity quantity to the gradient mapping (using Proposition~\ref{prop:relate-fos-to-grad-mapping}). 
\item Establish a first-order stationarity iteration complexity using the gradient mapping and the smoothness of the potential function (see Proposition~\ref{prop:fos-opt-guarantee}). This part of the proof is standard in non-convex smooth optimization. 
\end{enumerate}

\noindent\textbf{Proof.} In the remainder of this section, we provide a complete proof. 

We start by presenting a standard result for smooth non-convex optimization reported and proved in e.g. \cite[Lemma 12]{zhang-et-al24tac}. 

To state the result, we first define the standard gradient mapping to quantify first-order stationarity in our constrained setting: For any policy $\pi \in \Pi$ and any greediness policy parameter $\alpha >0$, we define:   
\begin{equation}
\label{eq:det-grad-mapping}
G^{\eta, \alpha}(\pi) := \frac{1}{\eta} \left( \text{Proj}(\pi + \eta \nabla \Phi(\pi)) - \pi \right)\,, 
\end{equation}
where $\Phi: \Pi \to \R$ is a given smooth function. Note that we will use the simpler notation 
\begin{equation}
G^{\eta}(\pi) := G^{\eta, \alpha}(\pi) \quad \text{when} \quad \alpha = 0\,,
\end{equation}
as it will be the case in the rest of this section since $\alpha$-greedy policies are not needed in the exact gradient setting. 

\begin{proposition}[e.g. Lemma 12, \cite{zhang-et-al24tac}]
\label{prop:fos-opt-guarantee}
Let $\Phi: \Pi \to \R$ be $\beta$-smooth and bounded, i.e. there exist~$\Phi_{\min}, \Phi_{\max} > 0$ s.t. for all $\pi \in \Pi, \Phi_{\min} \leq \Phi(\pi) \leq \Phi_{\max}\,.$ Then running projected gradient descent:
\begin{equation}
\pi^{t+1} = \text{Proj}_{\Pi}(\pi^t - \eta \nabla_{\pi} \Phi(\pi^t))\,,
\end{equation}
with stepsize $\eta \leq 1/\beta$ guarantees that $\lim_{t \to +\infty} \|G^{\eta}(\pi^t)\|_2 = 0\,.$ Moreover, we also have: 
\begin{equation}
\frac{1}{T} \sum_{t=1}^{T} \|G^{\eta}(\pi^t)\|_2^2 \leq \frac{2 \beta (\Phi_{\max} - \Phi_{\min})}{T}\,.
\end{equation}
\end{proposition}

The following result relates FOS with the gradient mapping norm in the deterministic setting.  

\begin{proposition}[Lemma 11, \cite{zhang-et-al24tac}]
\label{prop:relate-fos-to-grad-mapping}
Let $\Phi: \Pi \to \mathbb{R}$ be $\beta$-smooth for some constant $\beta > 0$ and let $\pi^{+} = \text{Proj}_{\Pi^\alpha}(\pi + \eta \nabla \Phi(\pi))$ for $\eta > 0\,.$ Then for all $\pi' \in \Pi^{\alpha}\,,$ 
\begin{equation}
\langle \nabla \Phi(\pi^{+}), \pi' - \pi^{+} \rangle \leq (1+\eta \beta) \|G^{\eta, \alpha}(\pi)\| \cdot \|\pi' - \pi^{+} \|\,.
\end{equation}
Moreover, we have 
\begin{equation}
\max_{\bar{\pi}_i \in \Pi_i} \langle \nabla_{\pi_i} \Phi(\pi^{+}), \bar{\pi}_i - \pi_i^{+} \rangle \leq 2(1+\eta \beta) \|G^{\eta, \alpha}(\pi)\| + \frac{2 N l_{\infty}}{(1-\gamma)^2} \alpha\,.
\end{equation}
\end{proposition}

The proof consists in connecting the Nash gap to FOS via agent-wise gradient domination and then using the link between FOS and the gradient mapping as well as the gradient mapping rate which follows from smoothness of the potential function.

Using Propositions~\ref{prop:fos-opt-guarantee} and~\ref{prop:relate-fos-to-grad-mapping} together with our agent-wise gradient domination (Proposition~\ref{prop:agentwise-grad-dom}), we are now ready to complete the proof. 

\noindent\textbf{End of Proof of Thm.~\ref{thm:iter-complexity-exact-PG}.} 
By definition of the Nash equilibrium gap and using agent-wise gradient domination (Proposition~\ref{prop:agentwise-grad-dom}), we have for every~$i \in \mathcal{N}, t \geq 0$: 
\begin{equation}
\text{NE-Gap}_i(\pi^{t+1}) \leq \frac{C_{\mathcal{G}}}{1-\gamma} \underset{\pi_i' \in \Pi_i}{\max} \langle \pi_i' - \pi_i^{t+1}, \nabla_{\pi_i} u_i(\pi^{t+1}) \rangle\,.
\end{equation}
Using the potential structure, we have $\nabla_{\pi_i} u_i(\pi^{t+1}) = \nabla_{\pii} \phi(\pi^{t+1})$ and we obtain: 
\begin{equation}
\label{eq:ne-gep-to-grad-pot}
\text{NE-Gap}_i(\pi^{t+1}) \leq \frac{C_{\mathcal{G}}}{1-\gamma} \underset{\pi_i' \in \Pi_i}{\max} \langle \pi_i' - \pi_i^{t+1}, \nabla_{\pii} \phi(\pi^{t+1})\rangle\,.
\end{equation}
Using now Proposition~\ref{prop:relate-fos-to-grad-mapping} with $\alpha = 0$ together with the above inequality and taking the maximum over $i \in \mathcal{N}$ yields: 
\begin{equation}
\label{eq:ne-gap-to-grad-map-det}
\text{NE-Gap}(\pi^{t+1}) \leq \frac{2C_{\mathcal{G}} (1+ \eta \beta)}{1-\gamma} \|G^{\eta}(\pi^{t+1})\|_2\,.
\end{equation}
To conclude, it remains to sum up the above inequalities over $t \in [T-1]$, use Jensen's inequality with the square root function and Proposition~\ref{prop:fos-opt-guarantee} to obtain: 
\begin{align}
\frac{1}{T}\sum_{t=0}^{T-1}\text{NE-Gap}(\pi^{t+1}) 
&\leq \frac{2C_{\mathcal{G}} (1+ \eta \beta)}{1-\gamma} \frac{1}{T}\sum_{t=0}^{T-1}\|G^{\eta}(\pi^{t+1})\|_2 \nonumber\\
&\leq \frac{2C_{\mathcal{G}} (1+ \eta \beta)}{1-\gamma} \sqrt{\frac{1}{T} \sum_{t=1}^{T}\|G^{\eta}(\pi^{t})\|_2^2}\nonumber\\
&\leq \frac{4C_{\mathcal{G}}}{1-\gamma} \sqrt{\frac{2 \beta (\Phi_{\max} - \Phi_{\min})}{T}}\,,
\end{align}
where the last step also uses the condition~$\eta \leq 1/\beta$. It follows that for any $\varepsilon >0$, if \begin{equation}
T \geq \frac{32 C_{\mathcal{G}}^2 \beta (\Phi_{\max} - \Phi_{\min})}{(1-\gamma)^2\varepsilon^2}\,,
\end{equation} 
then there exists $t \in [T]$ s.t. $\text{NE-Gap}(\pi^{t}) \leq \varepsilon\,.$

\subsection{Proof of Thm.~\ref{thm:sample-complexity-gen-model}}

\noindent\textbf{Proof sketch and technical comparison to prior work.} The proof follows similar lines as the proof of Theorem~\ref{thm:iter-complexity-exact-PG} but we now need to account for the errors due to inexactness of gradients used in the policy update rule. In particular the following steps deviate from the previous deterministic setting: 
\begin{itemize}
\item As we prove a high-probability result, the first-order stationarity guarantee now holds pathwise (i.e. without expectation) and has an additional error term due to inexact gradients (see Proposition~\ref{prop:fos-stoch-opt-guarantee}). 
\item Our policy gradient estimator in the generative model setting differs from both \cite{leonardos-et-al22iclr} (which does not consider the generative model setting) and \cite{zhang-et-al24tac}. Compared to the latter we rather use trajectory sampling with minibatches and we do not resort to model estimation (estimating the transition kernel). Our resulting sample complexity of~$\tilde{O}(\varepsilon^{-4})$ in the generative model setting improves over both the results of \cite{leonardos-et-al22iclr} and \cite{zhang-et-al24tac}. 
\item To control inexactness of gradients (see Proposition~\ref{prop:high-prob-grad-errors-bound}), we decompose the errors into a bias term of the order~$\gamma^H$ and a sampling error depending on the minibatch size~$M$ (using concentration inequalities). Then we select the minibatch size~$M$ to guarantee that the sampling error is of the same order as the bias term. This leads to a minibatch size of order $M = \tilde{O}(\gamma^{-2H})\,.$ It remains to choose the horizon~$H$ logarithmically in the inverse desired precision to control the bias and hence choose the corresponding minibatch size which is as a consequence of order~$\tilde{O}(\varepsilon^{-2})\,.$
\end{itemize} 

\noindent\textbf{Proof.} We now develop the proof in details. 

We introduce the stochastic gradient mapping similarly to its deterministic counterpart introduced in \eqref{eq:det-grad-mapping}. For any policy $\pi \in \Pi$ and any greediness policy parameter $\alpha \geq 0$, we define for all $t \geq 0$,   
\begin{equation}
\label{eq:stoch-grad-mapping}
\hat{G}^{\eta, \alpha}(\pi^t) := \frac{1}{\eta} \left( \text{Proj}_{\Pi^{\alpha}}(\pi + \eta g^t) - \pi \right)\,, 
\end{equation}
where for every $i \in [N], g_i^t := \hat{\nabla}_{\pi_i} u_{\mu,i}(\pi^t)\,.$ 

\begin{proposition}
\label{prop:fos-stoch-opt-guarantee}
Let $\Phi: \Pi \to \R$ be $\beta$-smooth and bounded, i.e. there exist~$\Phi_{\min}, \Phi_{\max} > 0$ s.t. for all $\pi \in \Pi, \Phi_{\min} \leq \Phi(\pi) \leq \Phi_{\max}\,.$ Let $(\pi^t)$ be the sequence of policies generated by running stochastic projected gradient descent:\footnote{In this section, we will only use $\alpha = 0$ given a generative model for sampling.}
\begin{equation}
\pi^{t+1} = \text{Proj}_{\Pi}(\pi^t - \eta g^t)\,, \quad g_i^t = \hat{\nabla}_{\pi_i} u_{\mu,i}(\pi^t)\,, \forall i \in [N]\,,
\end{equation}
with stepsize $\eta \leq 1/(1+\beta)$. If there exists~$\tau > 0$ s.t. for all $t, \|g^t - \nabla \Phi(\pi^t)\| \leq \tau$, then we have for all $t \geq 1$: 
\begin{equation}
\frac{1}{T}\sum_{t=1}^T \|\hat{G}^{\eta,\alpha}(\pi^t)\|_2^2 \leq \frac{2(\Phi_{\max} - \Phi_{\min})}{\eta T} + \frac{1}{\eta}\tau^2\,.
\end{equation}
\end{proposition}

\begin{proof}
By $\beta$-smoothness of~$\Phi$, we have for every time step $t \geq 0$: 
\begin{align}
\label{eq:smoothnes-ineq-gen}
\Phi(\pi^{t+1}) &\leq \Phi(\pi^{t}) + \langle \nabla \Phi(\pi^t), \pi^{t+1} - \pi^{t} \rangle + \frac{\beta}{2} \|\pi^{t+1} - \pi^{t}\|_2^2\,, \nonumber\\
&= \Phi(\pi^{t}) +  \langle g^t, \pi^{t+1} - \pi^{t} \rangle + \langle \nabla \Phi(\pi^t) - g^t, \pi^{t+1} - \pi^{t} \rangle + \frac{\beta}{2} \|\pi^{t+1} - \pi^{t}\|_2^2\,, \nonumber\\
&\leq \Phi(\pi^{t}) + \langle g^t, \pi^{t+1} - \pi^{t} \rangle +  \frac 12 \|\nabla \Phi(\pi^t) - g^t\|_2^2 + \frac{1+\beta}{2} \|\pi^{t+1} - \pi^{t}\|_2^2\,.
\end{align}
By characterization of the projection, the update rule of biased stochastic projected gradient descent yields:
\begin{equation}
\forall \pi \in \Pi, \langle \pi - \pi^{t+1} , (\pi^t - \eta g^t) - \pi^{t+1} \rangle \leq 0\,.
\end{equation}
Setting $\pi = \pi^t$ in the above inequality and rearranging, we obtain:
\begin{equation}
\label{eq:charac-proj-term-gen}
\langle g^t, \pi^{t+1} - \pi^t \rangle \leq - \frac{1}{\eta} \|\pi^{t+1} - \pi^t\|_2^2 = - \eta \|\hat{G}^{\eta,\alpha}(\pi^t)\|_2^2\,.
\end{equation}
Moreover, using our inexactness bound assumption, we also have: 
\begin{equation}
\label{eq:interm-bias-var-gen}
\frac 12 \|\nabla \Phi(\pi^t) - g^t\|_2^2 \leq \frac{1}{2}\tau^2\,.
\end{equation}

Using the definition of the (stochastic) gradient mapping together with \eqref{eq:interm-bias-var-gen} and \eqref{eq:charac-proj-term-gen} in \eqref{eq:smoothnes-ineq-gen}, we obtain: 
\begin{equation}
\Phi(\pi^{t+1}) \leq \Phi(\pi^{t}) - \eta \left(1 - \frac{\eta (1+ \beta)}{2}\right) \|\hat{G}^{\eta,\alpha}(\pi^t)\|_2^2 + \frac{1}{2}\tau^2\,.
\end{equation}
Choosing $\eta \leq \frac{1}{1+\beta}$ gives $1 - \frac{\eta (1+ \beta)}{2} \geq \frac 12$ and we obtain by rearranging the above inequality: 
\begin{equation}
\|\hat{G}^{\eta,\alpha}(\pi^t)\|_2^2 \leq \frac{2}{\eta} (\Phi(\pi^{t}) - \Phi(\pi^{t+1})) + \frac{1}{\eta}\tau^2\,.
\end{equation}
Summing the above inequality for $t \in [T]$ yields: 
\begin{equation}
\frac{1}{T}\sum_{t=1}^T \|\hat{G}^{\eta,\alpha}(\pi^t)\|_2^2 \leq \frac{2(\Phi_{\max} - \Phi_{\min})}{\eta T} + \frac{1}{\eta}\tau^2\,.
\end{equation}
\end{proof}

\begin{proposition}[Lemma 11 in \cite{zhang-et-al24tac}]
\label{prop:relate-fos-to-grad-mapping-stoch}
Let $\Phi: \Pi \to \R$ be $\beta$-smooth and bounded, i.e. there exist~$\Phi_{\min}, \Phi_{\max} > 0$ s.t. for all $\pi \in \Pi, \Phi_{\min} \leq \Phi(\pi) \leq \Phi_{\max}\,.$ Let $(\pi^t)$ be the sequence of policies generated by running stochastic projected gradient descent:
\begin{equation}
\pi^{t+1} = \text{Proj}_{\Pi}(\pi^t - \eta g^t)\,, \quad g_i^t = \hat{\nabla}_{\pi_i} u_{\mu,i}(\pi^t)\,, \forall i \in [N]\,,
\end{equation}
with stepsize $\eta \leq 1/(1+\beta)$. If there exists~$\tau > 0$ s.t. for all $t, \|g^t - \nabla \Phi(\pi^t)\|_2 \leq \tau$, then we have for all $t \geq 1$: 
\begin{equation}
\underset{i \in \mathcal{N}}{\max} \underset{\pi_i' \in \Pi_i}{\max} \langle \pi_i' - \pi_i^{t+1}, \nabla_{\pi_i} \Phi(\pi^{t+1}) \rangle \leq 2 \sqrt{|\cS|}(1+\eta \beta) \|\hat{G}^{\eta, \alpha}(\pi^t)\|  + 2 \sqrt{|\cS|} \tau\,.
\end{equation}
\end{proposition}

We now quantify the magnitude~$\tau$ of the gradient errors using concentration results for depending on the minibatch size~$M$ used for sampled trajectories and the truncation horizon~$H$. 

\begin{proposition}
\label{prop:high-prob-grad-errors-bound}
For any~$\delta \in (0,1),$ if $M \geq 2 \gamma^{-2H} \log(\frac{2 |\cS| \cdot \max_{k \in [N]}|\cA_k|}{\delta}),$ then with probability at least $1-3\delta$,
\begin{equation}
\|\hat{\nabla}_{\pi_i} u_{\mu, i}(\pi^t) - \nabla_{\pi_i} u_{\mu, i}(\pi^t)\|_{\infty} \leq \frac{2 N}{1-\gamma} \left[ l_{\infty} \left( 1+ \frac{1}{1-\gamma}\right) + N L\right] \gamma^{H}\,.
\end{equation}
\end{proposition}

\begin{proof}
See section~\ref{sec:proof-high-prob-grad-errors-bound} below. 
\end{proof}

\noindent\textbf{End of Proof of Thm.~\ref{thm:sample-complexity-gen-model}.} 
Similarly to the proof of the deterministic setting (Thm.~\ref{thm:iter-complexity-exact-PG}, see \eqref{eq:ne-gep-to-grad-pot}), we have: 
\begin{equation}
\label{eq:ne-gep-to-grad-pot-sto}
\text{NE-Gap}_i(\pi^{t+1}) \leq \frac{C_{\mathcal{G}}}{1-\gamma} \underset{\pi_i' \in \Pi_i}{\max} \langle \pi_i' - \pi_i^{t+1}, \nabla_{\pii} \phi(\pi^{t+1})\rangle\,.
\end{equation}

Using Proposition~\ref{prop:relate-fos-to-grad-mapping-stoch} with $\alpha = 0$ together with the above inequality and taking the maximum over $i \in \mathcal{N}$ yields: 
\begin{equation}
\label{eq:ne-gap-to-grad-map-det}
\text{NE-Gap}(\pi^{t+1}) \leq \frac{2 C_{\mathcal{G}} \sqrt{|\cS|}}{1-\gamma} \left((1+ \eta \beta) \|\hat{G}^{\eta, \alpha}(\pi^t)\|_2 + \tau\right)\,.
\end{equation}
Summing up the above inequalities over $t \in \{1, \ldots, T\}$, using Jensen's inequality with the square root function and Proposition~\ref{prop:fos-stoch-opt-guarantee}, we obtain: 
\begin{align}
\frac{1}{T}\sum_{t=1}^{T}\text{NE-Gap}(\pi^{t+1}) 
&\leq \frac{2 C_{\mathcal{G}} \sqrt{|\cS|}}{1-\gamma} \frac{1}{T}\sum_{t=1}^{T}\left((1+ \eta \beta) \|\hat{G}^{\eta, \alpha}(\pi^t)\|_2 + \tau\right) \nonumber\\
&\leq \frac{2 C_{\mathcal{G}} \sqrt{|\cS|}}{1-\gamma} \left(\tau + \frac{1+ \eta \beta}{T}\sum_{t=1}^{T} \|\hat{G}^{\eta, \alpha}(\pi^t)\|_2 \right) \nonumber\\
&\leq \frac{2 C_{\mathcal{G}} \sqrt{|\cS|}}{1-\gamma} \left(\tau + (1+ \eta \beta) \sqrt{\frac{1}{T}\sum_{t=1}^{T} \|\hat{G}^{\eta, \alpha}(\pi^t)\|_2^2} \right)\nonumber\\
&\leq \frac{2 C_{\mathcal{G}} \sqrt{|\cS|}}{1-\gamma} \left(\tau + (1+ \eta \beta) \sqrt{\frac{2(\Phi_{\max} - \Phi_{\min})}{\eta T} + \frac{1}{\eta}\tau^2} \right)\nonumber\\
&\leq \frac{2 C_{\mathcal{G}} \sqrt{|\cS|}}{1-\gamma} \left( \left(1+ \frac{2}{\sqrt{\eta}}\right)\tau + \sqrt{\frac{8(\Phi_{\max} - \Phi_{\min})}{\eta T}} \right)\,,
\end{align}
where the last step also uses the condition~$\eta \leq 1/(1+\beta)$. It follows that for any $\varepsilon >0$, if 
\begin{equation}
\label{eq:nb-iter-selection-gen}
T \geq \frac{128 (1+\beta) C_{\mathcal{G}}^2 |\cS| (\Phi_{\max} - \Phi_{\min})}{(1-\gamma)^2} \cdot\frac{1}{\varepsilon^2}\,, \quad \tau \leq \frac{1-\gamma}{4(1+ 2 \sqrt{1+\beta}) C_{\mathcal{G}}\sqrt{|\cS|}} \cdot \varepsilon\,,
\end{equation} 
then, we have $\frac{1}{T}\sum_{t=1}^{T}\text{NE-Gap}(\pi^{t+1}) \leq \varepsilon\,.$ 
Hence, there exists $t \in [T]$ s.t. $\text{NE-Gap}(\pi^{t}) \leq \varepsilon\,.$ 

It remains to select the minibatch size~$M$ and the truncation horizon~$H$ appropriately using Proposition~\ref{prop:high-prob-grad-errors-bound} and the condition on $\tau$ above which gives: 
\begin{equation}
C_1 \gamma^H \leq C_2 \varepsilon\,,
\end{equation}
where the constants $C_1, C_2$ are independent of $H$ and defined as follows: 
\begin{equation}
\label{eq:def-C1-C2}
C_1:= \frac{2 N}{1-\gamma} \left[ l_{\infty} \left( 1+ \frac{1}{1-\gamma}\right) + N L\right], \quad C_2 := \frac{1-\gamma}{4(1+ 2 \sqrt{1+\beta}) C_{\mathcal{G}}\sqrt{|\cS|}}\,. 
\end{equation}
It suffices to choose $H$ such that $\exp(- \frac{H}{2} (1-\gamma)) \leq \frac{C_2}{C_1} \varepsilon\,,$ i.e. select $H$ such that: 
\begin{equation}
\label{eq:horizon-selection-gen}
H \geq \frac{2}{1-\gamma} \log\left(\frac{C_2}{C_1} \cdot \frac{1}{\varepsilon} \right)\,, 
\end{equation}
where $C_2, C_1$ are defined in \eqref{eq:def-C1-C2}. Then, using Proposition~\ref{prop:high-prob-grad-errors-bound}, we obtain the minibatch size condition: 
\begin{equation}
\label{eq:minibatch-selection-gen}
M \geq 2 \gamma^{-2H} \log\left(\frac{2 |\cS| \cdot \max_{k \in [N]}|\cA_k|}{\delta}\right) \geq 2 \left(\frac{C_1}{C_2} \right)^2 \frac{1}{\varepsilon^2} \log\left(\frac{2 |\cS| \cdot \max_{k \in [N]}|\cA_k|}{\delta}\right),
\end{equation}
since $\gamma^H \leq \frac{C_2}{C_1} \varepsilon\,.$ Overall the total sample complexity is given by: 
\begin{equation}
T \times M \times H = \tilde{O}\left(\varepsilon^{-2} \times \varepsilon^{-2} \times \log\left(\frac{1}{\varepsilon}\right) \right) = \tilde{O}(\varepsilon^{-4})\,,
\end{equation}
where all the constants are specified above in the choice of $T, M, H$ respectively in \eqref{eq:nb-iter-selection-gen}, \eqref{eq:minibatch-selection-gen} and \eqref{eq:horizon-selection-gen}. This concludes the proof.

\subsection{Proof of Proposition~\ref{prop:high-prob-grad-errors-bound}}
\label{sec:proof-high-prob-grad-errors-bound}

We use the policy gradient theorem (Theorem~\ref{prop:individual-pg}) to write: 
\begin{align}
\label{eq:decomp-grad-error-proof}
&\|\hat{\nabla}_{\pi_i} u_{\mu, i}(\pi^t) - \nabla_{\pi_i} u_{\mu, i}(\pi^t)\|_{\infty} 
= \max_{s \in \cS, a_i \in \cA_i} \left|\frac{\hat{\partial} u_{\mu,i}(\pi^t)}{\partial \pi_i(a_i|s)} - \frac{\partial u_{\mu,i}(\pi^t)}{\partial \pi_i(a_i|s)} \right| \nonumber\\
&= \frac{1}{1-\gamma} \max_{s \in \cS, a_i \in \cA_i} \left|\hat{d}_{\mu}^{\pi^t}(s) \sum_{j=1}^N \hat{q}^{\pi^t}_{i,j, (s,a_i)} -  d_{\mu}^{\pi^t}(s) \sum_{j=1}^N \bar{Q}_{s,a_i}^{\pi^t}(\bar{r}_{i,j}^{\pi^t})\right| \nonumber\\
&\leq \frac{1}{1-\gamma} \max_{s \in \cS, a_i \in \cA_i} \left(\hat{d}_{\mu}^{\pi^t}(s) \cdot \sum_{j=1}^N \|\hat{q}_j^{\pi^t} - \bar{Q}^{\pi^t}(\bar{r}_{i,j}^{\pi^t})\|_{\infty} + |\sum_{j=1}^N \bar{Q}_{s,a_i}^{\pi^t}(\bar{r}_{i,j}^{\pi^t})| \cdot |\hat{d}_{\mu}^{\pi^t}(s) - d_{\mu}^{\pi^t}(s)| \right) \nonumber\\
&\leq \frac{1}{1-\gamma} \left( \sum_{j=1}^N \|\hat{q}_j^{\pi^t} - \bar{Q}^{\pi^t}(\bar{r}_{i,j}^{\pi^t})\|_{\infty} + \frac{Nl_{\infty}}{1-\gamma} \|\hat{d}_{\mu}^{\pi^t} - d_{\mu}^{\pi^t}\|_{\infty} \right)\,.
\end{align}
In what follows, we establish high probability error bounds for each one of the error terms in the above inequality. 

\begin{lemma}
\label{lem:concentration-state-occup}
For any $\delta \in (0,1),$ if $M \geq \frac{\gamma^{-2H}}{2} \log(\frac{2 |\cS|}{\delta}),$ then with probability at least $1-\delta,$ 
\begin{equation}
\|\hat{d}_{\mu}^{\pi^t} - d_{\mu}^{\pi^t}\|_{\infty} \leq 2 \gamma^H\,,
\end{equation}
for all iterations~$t \geq 1\,.$
\end{lemma}

\begin{proof}
First, we upper bound the error into two terms: a bias error due to horizon truncation and a concentration error which vanishes when increasing the number of sampled trajectories: 
\begin{equation}
\|\hat{d}_{\rho}^{\pi^t} - d_{\rho}^{\pi^t}\|_{\infty} \leq \|\hat{d}_{\rho}^{\pi^t} - \bE[\hat{d}_{\rho}^{\pi^t}]\| + \|\bE[\hat{d}_{\rho}^{\pi^t}] - d_{\rho}^{\pi^t}\|_{\infty}\,.
\end{equation}
For the second error, we observe that:
\begin{equation}
\label{eq:bias-d}
\|\bE[\hat{d}_{\rho}^{\pi^t}] - d_{\rho}^{\pi^t}\|_{\infty} = \max_{s \in \cS} \left| \sum_{t=H}^{+\infty} (1-\gamma) \gamma^t \mathbb{P}(s_t^{(j)}= s) \right|
\leq (1-\gamma) \sum_{t=H}^{+\infty} \gamma^t = \gamma^H\,.
\end{equation}

As for the first error term, we use Hoeffding's inequality. For any state $s \in \cS$ and any $\sigma > 0$, since $\hat{d}_{\rho}^{\pi^t}(s)$ is bounded (in $[0,1]$), we have: 
\begin{equation}
\mathbb{P}(|\hat{d}_{\rho}^{\pi^t}(s) - \bE[\hat{d}_{\rho}^{\pi^t}(s)]| \geq \sigma) \leq 2 \exp(-2 M \sigma^2)\,,
\end{equation}
where~$M$ is the minibatch size used in the Monte Carlo estimate of~$\hat{d}_{\rho}^{\pi^t}(s)\,.$ 
Using a union bound, we obtain: 
\begin{equation}
\mathbb{P}(\|\hat{d}_{\rho}^{\pi^t}(s) - \bE[\hat{d}_{\rho}^{\pi^t}(s)]\|_{\infty} \geq \sigma)
\leq \mathbb{P}\left( \cup_{s \in \cS} \{|\hat{d}_{\rho}^{\pi^t}(s) - \bE[\hat{d}_{\rho}^{\pi^t}(s)]| \geq \sigma\} \right) \leq 2 |\cS| \exp(-2 M \sigma^2)\,.
\end{equation}
Hence, for any $\delta \in (0,1),$ if $M \geq \frac{1}{2\sigma^2} \log(\frac{2 |\cS|}{\delta}),$ then with probability at least $1-\delta,$ we have with probability at least $1-\delta$, 
\begin{equation}
\label{eq:hoeff-d}
\|\hat{d}_{\rho}^{\pi^t}(s) - \bE[\hat{d}_{\rho}^{\pi^t}(s)]\|_{\infty} \leq \sigma\,.
\end{equation}
Combining \eqref{eq:bias-d} and \eqref{eq:hoeff-d} yields: 
\begin{equation}
\|\hat{d}_{\rho}^{\pi^t} - d_{\rho}^{\pi^t}\|_{\infty} \leq \gamma^H + \sigma\,,
\end{equation}
with probability at least $1-\delta\,.$ It remainss to set $\sigma = \gamma^H$ to obtain the desired result.
\end{proof}

\begin{lemma}
\label{lem:concentration-q-values}
For any $\delta \in (0,1),$ if $M \geq 2\gamma^{-2H} \log(\frac{2 |\cS| \cdot \max_{k \in [N]}|\cA_k|}{\delta}),$ then with probability at least $1-2\delta,$ 
\begin{equation}
\|\hat{q}_j^{\pi^t} - \bar{Q}^{\pi^t}(\bar{r}_{i,j}^{\pi^t})\|_{\infty} \leq \left(\frac{2l_{\infty}}{1-\gamma} + 2NL \right) \gamma^H\,,
\end{equation}
for all iterations~$t \geq 1$ and for all $i,j \in [N]\,.$ 
\end{lemma}

\begin{proof}
We decompose the Q-value estimation errors as follows: 
\begin{align}
\label{eq:decomp-error-q-val-diff}
\|\hat{q}_j^{\pi^t} - \bar{Q}^{\pi^t}(\bar{r}_{i,j}^{\pi^t})\|_{\infty} 
&= \|\hat{q}_j^{\pi^t}(\hat{r}_{i,j}^{\pi^t}) - \bar{Q}^{\pi^t}(\bar{r}_{i,j}^{\pi^t})\|_{\infty} \nonumber\\
&= \|\hat{q}_j^{\pi^t}(\hat{r}_{i,j}^{\pi^t}) - \hat{q}_j^{\pi^t}(\bar{r}_{i,j}^{\pi^t}) 
+ \hat{q}_j^{\pi^t}(\bar{r}_{i,j}^{\pi^t}) - \bar{Q}^{\pi^t}(\bar{r}_{i,j}^{\pi^t})\|_{\infty}\nonumber\\
&\leq \|\hat{q}_j^{\pi^t}(\hat{r}_{i,j}^{\pi^t}) - \hat{q}_j^{\pi^t}(\bar{r}_{i,j}^{\pi^t})\|_{\infty} 
+ \|\hat{q}_j^{\pi^t}(\bar{r}_{i,j}^{\pi^t}) - \bar{Q}^{\pi^t}(\bar{r}_{i,j}^{\pi^t})\|_{\infty}\,.
\end{align}

For the first term, we observe that: 
\begin{equation}
\|\hat{q}_j^{\pi^t}(\hat{r}_{i,j}^{\pi^t}) - \hat{q}_j^{\pi^t}(\bar{r}_{i,j}^{\pi^t})\|_{\infty} 
\leq \frac{1}{1-\gamma} \|\hat{r}_{i,j}^{\pi^t} - \bar{r}_{i,j}^{\pi^t}\|_{\infty, \tau_{1:M}} 
\leq L \|\hat{\lambda}_{\mu,1:N}^{\pi^t} - \lambda_{\mu,1:N}^{\pi^t}\|_{1} 
\leq L \sum_{k=1}^N \|\hat{\lambda}_{\mu,k}^{\pi^t} - \lambda_{\mu,k}^{\pi^t}\|_{1} 
= N L \|\hat{d}_{\mu}^{\pi^t} - d_{\mu}^{\pi^t}\|_1\,.
\end{equation}
where the second inequality follows from using Assumption~\ref{as:smoothness} together with the fact that there is no dependence on~$\pi_{-i}^t$ in the pseudo-rewards in the potential GUMG setting. 

We can now apply Lemma~\ref{lem:concentration-state-occup} to obtain that, if $M \geq \frac{\gamma^{-2H}}{2} \log(\frac{2 |\cS|}{\delta})$, then with probability at least $1-\delta,$ 
\begin{equation}
\|\hat{q}_j^{\pi^t}(\hat{r}_{i,j}^{\pi^t}) - \hat{q}_j^{\pi^t}(\bar{r}_{i,j}^{\pi^t})\|_{\infty} \leq 2 NL \gamma^H\,. 
\end{equation} 

As for the second term in \eqref{eq:decomp-error-q-val-diff}, we use the same proof technique as the proof of Lemma~\ref{lem:concentration-state-occup}, combining a bias term bound and an application of Hoeffding's concentration inequality. Observe for this that the random variables $\hat{q}_j^{\pi^t}(\hat{r}_{i,j}^{\pi^t})$ are bounded (in $[\frac{-l_{\infty}}{1-\gamma},\frac{l_{\infty}}{1-\gamma}]$) by boundedness of the pseudo-rewards and use $\sigma = \frac{l_{\infty}}{1-\gamma} \gamma^H$. Therefore, following the same lines we obtain that for any $\delta \in (0,1),$ if $M \geq 2\gamma^{-2H} \log(\frac{2 |\cS| \cdot \max_{k \in [N]} |\cA_k|}
{\delta}),$ then with probability at least $1-\delta,$ 
\begin{equation}
\|\hat{q}_j^{\pi^t}(\bar{r}_{i,j}^{\pi^t}) - \bar{Q}^{\pi^t}(\bar{r}_{i,j}^{\pi^t})\|_{\infty} \leq \frac{2 l_{\infty}}{1-\gamma} \gamma^H\,. 
\end{equation}
\end{proof}

Combining the results of Lemma~\ref{lem:concentration-state-occup} and Lemma~\ref{lem:concentration-q-values} in \eqref{eq:decomp-grad-error-proof} gives the desired inequality and concludes the proof of Proposition~\ref{prop:high-prob-grad-errors-bound}.  

\subsection{Proof of Thm.~\ref{thm:sample-complexity-on-policy}}

\noindent\textbf{Proof sketch and comparison to prior work.} Our proof technique here differs from the proof of a similar result of \cite{leonardos-et-al22iclr} (Theorem 4.4) for the particular case of Markov potential games. Note that we do not use the Moreau-envelope in our analysis for this stochastic setting. We rather provide an analysis of projected stochastic gradient descent with minibatch, controlling the bias and variance of the policy gradient estimates. The following decomposition gives more intuition on the different errors involved in the analysis: 
\begin{equation}
\frac{1}{T} \sum_{t=0}^{T-1} \E[\text{NE-Gap}(\pi^{t+1})] 
= \underbrace{\tilde{O}\left(\frac{1}{\sqrt{T}}\right)}_{\text{Optimization error}} + \underbrace{\tilde{O}\left(\gamma^H\right)}_{\text{bias due to truncation}} + \underbrace{\tilde{O}\left(\frac{1}{\sqrt{M\alpha}}\right)}_{\text{variance term}} + \underbrace{\tilde{O}\left(\alpha\right)}_{\text{bias due to $\alpha$-greediness}}\,.
\end{equation}
This suggests to choose $T = \tilde{O}(\varepsilon^{-2})$ (like in the deterministic setting), a greediness parameter $\alpha =  \tilde{O}(\varepsilon)$, a minibatch size $M =  \tilde{O}(\varepsilon^{-3})$ and a horizon length $H = O\left(\frac{1}{1-\gamma} \log\left(\frac{1}{\varepsilon}\right)\right)\,.$ The resulting total sample complexity is: $TMH =  \tilde{O}(\varepsilon^{-5})\,.$ 

\noindent\textbf{Proof.} We provide below a complete proof. We start with a first-order stationarity guarantee for smooth non-convex optimization with biased stochastic gradients.  

\begin{proposition}
\label{prop:fos-stoch-opt-guarantee-on-pol}
Let $\Phi: \Pi \to \R$ be $\beta$-smooth and bounded, i.e. there exist~$\Phi_{\min}, \Phi_{\max} > 0$ s.t. for all $\pi \in \Pi, \Phi_{\min} \leq \Phi(\pi) \leq \Phi_{\max}\,.$ Let $(\pi^t)$ be the sequence of policies generated by running stochastic projected gradient descent:
\begin{equation}
\pi^{t+1} = \text{Proj}_{\Pi^{\alpha}}(\pi^t - \eta g^t)\,, \quad g_i^t = \hat{\nabla}_{\pi_i} u_{\mu,i}(\pi^t)\,, \forall i \in [N]\,,
\end{equation}
with stepsize $\eta \leq 1/(1+\beta)$. Suppose that the bias and variance of the policy gradient estimator~$g^t$ are bounded by positive constants~$\delta, \sigma^2$ as follows: 
\begin{equation}
\label{as:bound-bias-var}
\|\E_t[g^t] - \nabla \Phi(\pi^t)\|_{2} \leq \delta\,,
\quad 
\E\left[ \|g^t - \E_t[g^t]\|_{2}^2 \right] \leq \sigma^2\,,
\end{equation}
where~$\mathbb{E}_t$ is the conditional expectation w.r.t. randomness up to time~$t$. 
Then we have for any~$\alpha \geq 0, t \geq 1$: 
\begin{equation}
\frac{1}{T} \sum_{t=1}^{T} \E[\|\hat{G}^{\eta,\alpha}(\pi^t)\|_2^2] \leq \frac{2 (\Phi_{\max} - \Phi_{\min})}{\eta T} + \delta^2 + \sigma^2\,,
\end{equation}
where we recall that~$\hat{G}^{\eta, \alpha}(\pi^t) =  \frac{1}{\eta} \left( \text{Proj}_{\Pi^{\alpha}}(\pi^t + \eta g^t) - \pi^t \right)$ for any $\alpha \geq 0, t \geq 1\,.$ 
\end{proposition}

\begin{proof}
See section~\ref{sec:proof-prop:fos-stoch-opt-guarantee-on-pol}
\end{proof}

\begin{lemma}[Bias bound]
\label{lem:bias-pg}
Under Assumption~\ref{as:smoothness}, we have for any $i \in \mathcal{N}$:  
\begin{equation}
\E\left[ \|\nabla_{\pi} u_{\mu,i,H}(\pi) - \nabla_{\pi} u_{\mu,i}(\pi)\|_2 \right]
\leq \left(\frac{1+l_{\infty}}{1-\gamma} + NL \right)|\cS|\left(\sum_{k=1}^N|\cA_k|\right) \gamma^H\,,
\end{equation}
where $\nabla_{\pi} u_{\mu,i,H}(\pi)$ is a notation for the (conditional) expectation of the truncated policy gradient estimator (as in Proposition~\ref{prop:fos-stoch-opt-guarantee-on-pol}). 
\end{lemma}

\begin{proof}
The proof follows similar lines to the proof of Proposition~\ref{prop:high-prob-grad-errors-bound} where the bias is controlled in the proof of Lemma~\ref{lem:concentration-state-occup} (see eq.~\eqref{eq:bias-d}) and Lemma~\ref{lem:concentration-q-values}.  
\end{proof}

\begin{lemma}[Variance bound] 
\label{lem:var-pg}
For any $i \in \mathcal{N}$ and any policy~$\pi \in \Pi$, 
\begin{equation}
\E\left[ \|\hat{\nabla}_{\pi_i} u_{\mu,i}(\pi) - \nabla_{\pi_i} u_{\mu,i,H}(\pi)\|_2^2 \right] \leq \frac{N^2 l_{\infty}^2}{(1-\gamma)^4 M \alpha}\,.
\end{equation}
\end{lemma}

\begin{proof}
The proof follows similar lines to the proof of \cite[Lemma 2]{daskalakis-foster-golowich20}. It is enough to observe that the sum over $j$ of pseudo-rewards are bounded by $N l_{\infty}$ and our minibatch size~$M$ scales the variance. Then the variance for each trajectory estimator term is bounded using \cite[Lemma 2]{daskalakis-foster-golowich20}. 
\end{proof}

\begin{proposition}
\label{prop:relate-fos-to-grad-mapping-stoch}
Let $\Phi: \Pi \to \R$ be $\beta$-smooth and bounded, i.e. there exist~$\Phi_{\min}, \Phi_{\max} > 0$ s.t. for all $\pi \in \Pi, \Phi_{\min} \leq \Phi(\pi) \leq \Phi_{\max}\,.$ Let $(\pi^t)$ be the sequence of policies generated by running stochastic projected gradient descent:
\begin{equation}
\pi^{t+1} = \text{Proj}_{\Pi^{\alpha}}(\pi^t - \eta g^t)\,, \quad g_i^t = \hat{\nabla}_{\pi_i} u_{\mu,i}(\pi^t)\,, \forall i \in [N]\,,
\end{equation}
with stepsize $\eta \leq 1/(1+\beta)$. Suppose that the bias and variance of the policy gradient estimator~$g^t$ are bounded by positive constants~$\delta, \sigma^2$ as follows: 
\begin{equation}
\label{as:bound-bias-var}
\|\E_t[g^t] - \nabla \Phi(\pi^t)\|_{2} \leq \delta\,,
\quad 
\E\left[ \|g^t - \E_t[g^t]\|_{2}^2 \right] \leq \sigma^2\,,
\end{equation}
where~$\mathbb{E}_t$ is the conditional expectation w.r.t. randomness up to time~$t$. 
Then we have for any~$\alpha \geq 0,t \geq 1$: 
\begin{equation}
\E\left[\underset{i \in \mathcal{N}}{\max} \underset{\pi_i' \in \Pi_i}{\max} \langle \pi_i' - \pi_i^{t+1}, \nabla_{\pi_i} \Phi(\pi^{t+1}) \rangle \right] \leq 2 \sqrt{|\cS|}(1+\eta \beta) \E[\|\hat{G}^{\eta, \alpha}(\pi^t)\|_2] + 2 \sqrt{|\cS|} (\delta + \sigma) + \frac{2 N l_{\infty}}{(1-\gamma)^2} \alpha\,.
\end{equation}  
\end{proposition}

\begin{proof}
See section~\ref{sec:prop-relate-fos-to-grad-mapping-stoch-proof}. 
\end{proof}

The proof follows similar lines as the proof of \cite[Lemma 11]{zhang-et-al24tac} with the following differences: 
\begin{itemize}
\item Our setting involves general utilities beyond standard RL. This implies a few modifications needed including bounding the pseudo-rewards and using a suitable policy gradient theorem for our setting. 
\item We control the bias and variance of our estimator differently. Indeed, in this section, we are interested in deriving (on-policy) guarantees in expectation rather than with high probability.  
\end{itemize}

\noindent\textbf{End of Proof of Thm.~\ref{thm:sample-complexity-on-policy}.}  
Similarly to the proof of Thm.~\ref{thm:iter-complexity-exact-PG} with exact gradients, we start by bounding the Nash equilibrium gap using agent-wise gradient domination. We have for every $i \in \mathcal{N}$ and every time step~$t \geq 0$:
\begin{equation}
\text{NE-Gap}_i(\pi^{t+1}) \leq \frac{C_{\mathcal{G}}}{1-\gamma} \underset{\pi_i' \in \Pi_i}{\max} \langle \pi_i' - \pi_i^{t+1}, \nabla_{\pi_i} \Phi(\pi^{t+1}) \rangle\,.
\end{equation}
Taking the maximum over $i \in \mathcal{N}$ on both sides of the inequality, taking expectation and using Proposition~\ref{prop:relate-fos-to-grad-mapping-stoch}, we obtain 
\begin{equation}
\E[\text{NE-Gap}(\pi^{t+1})] \leq \frac{C_{\mathcal{G}}}{1-\gamma}\left[ 2 \sqrt{|\cS|}(1+\eta \beta) \E[\|\hat{G}^{\eta, \alpha}(\pi^{t+1})\|_2] + 2 \sqrt{|\cS|} (\delta + \sigma) + \frac{2 N l_{\infty}}{(1-\gamma)^2} \alpha \right]\,.
\end{equation}

Taking the average of the above inequality over $t \in [T-1]$ and using a stepsize $\eta \leq 1/(1+\beta)$, it follows that: 
\begin{align}
\label{eq:master-ineq-proof-thm4}
\frac{1}{T} \sum_{t=0}^{T-1} \E[\text{NE-Gap}(\pi^{t+1})] 
&\leq  \frac{C_{\mathcal{G}}}{1-\gamma}\left[ 4 \sqrt{|\cS|} \frac{1}{T} \sum_{t=0}^{T-1}  \E[\|\hat{G}^{\eta, \alpha}(\pi^{t+1})\|_2]  + 2 \sqrt{|\cS|} (\delta + \sigma) + \frac{2 N l_{\infty}}{(1-\gamma)^2} \alpha \right] \nonumber\\
&\leq \frac{C_{\mathcal{G}}}{1-\gamma}\left[   4 \sqrt{|\cS|} \sqrt{\frac{1}{T} \sum_{t=0}^{T-1}  \E[\|\hat{G}^{\eta, \alpha}(\pi^{t+1})\|_2^2]}  + 2 \sqrt{|\cS|} (\delta + \sigma) + \frac{2 Nl_{\infty}}{(1-\gamma)^2} \alpha \right] \nonumber\\
&\leq \frac{C_{\mathcal{G}}}{1-\gamma}\left[   4 \sqrt{|\cS|} \sqrt{\frac{2 (\Phi_{\max} - \Phi_{\min})}{\eta T} + \delta^2 + \sigma^2}  + 2 \sqrt{|\cS|} (\delta + \sigma) + \frac{2 Nl_{\infty}}{(1-\gamma)^2} \alpha \right]\nonumber\\
&\leq \frac{C_{\mathcal{G}}}{1-\gamma}\left[   4 \sqrt{|\cS|} \sqrt{\frac{2 (\Phi_{\max} - \Phi_{\min})}{\eta T}}  + 6 \sqrt{|\cS|} (\delta + \sigma) + \frac{2 Nl_{\infty}}{(1-\gamma)^2} \alpha \right]\nonumber\\
\end{align}
where the second inequality follows from Jensen's inequality and the third one from using Proposition~\ref{prop:fos-stoch-opt-guarantee-on-pol}. 

It remains to plug in the bias~$\delta$ and variance~$\sigma^2$ obtained in Lemma~\ref{lem:bias-pg} and Lemma~\ref{lem:var-pg} respectively, and select the minibatch size~$M$, the horizon length~$H$ and the greediness parameter~$\alpha$ appropriately to obtain the desired sample complexity.  

Using Lemma~\ref{lem:bias-pg} and Lemma~\ref{lem:var-pg} in \eqref{eq:master-ineq-proof-thm4}, we obtain: 
\begin{equation}
\frac{1}{T} \sum_{t=0}^{T-1} \E[\text{NE-Gap}(\pi^{t+1})] 
\leq \underbrace{\frac{\tilde{C}_1}{\sqrt{T}}}_{\text{Optimization error}} + \underbrace{\tilde{C}_2 \gamma^H}_{\text{bias due to truncation}} + \underbrace{\frac{\tilde{C}_3}{\sqrt{M\alpha}}}_{\text{variance term}} + \underbrace{\tilde{C}_4 \alpha}_{\text{bias due to $\alpha$-greediness}}\,,
\end{equation}
where the constants~$\tilde{C}_1, \tilde{C}_2, \tilde{C}_3, \tilde{C}_4$ are independent of~$T, M, H, \alpha$ and defined as follows: 
\begin{align}
\label{eq:def-constants-tilde-C}
\tilde{C}_1 &:= \frac{4 \sqrt{|\cS|} C_{\mathcal{G}}}{1-\gamma} \sqrt{\frac{2 (\Phi_{\max} - \Phi_{\min})}{\eta}}\,, 
\quad \tilde{C}_2:= \frac{6 C_{\mathcal{G}}\sqrt{|\cS|}}{1-\gamma} \left(\frac{1+l_{\infty}}{1-\gamma} + NL \right)|\cS|\left(\sum_{k=1}^N|\cA_k|\right)\,, \nonumber\\
\tilde{C}_3 &:= \frac{6 C_{\mathcal{G}}\sqrt{|\cS|} N^{\frac{3}{2}} l_{\infty}}{(1-\gamma)^3}\,, 
\quad \tilde{C}_4 :=\frac{2 C_{\mathcal{G}}N l_{\infty}}{(1-\gamma)^3}\,.
\end{align}
Selecting then $T, H, M, \alpha$ as follows: 
\begin{equation}
T \geq \frac{\tilde{C}_1^2}{\varepsilon^2}\,, \quad \alpha \leq \frac{\varepsilon}{\tilde{C}_4}\,, \quad M \geq \left(\frac{\tilde{C}_3}{\varepsilon\sqrt{\alpha}}\right)^2 = \frac{\tilde{C}_3^2 \tilde{C}_4}{\varepsilon^3}\,, \quad H \geq \frac{1}{1-\gamma} \log\left(\frac{\tilde{C}_2}{\varepsilon}\right)\,,
\end{equation}
guarantees that $\frac{1}{T} \sum_{t=0}^{T-1} \E[\text{NE-Gap}(\pi^{t+1})] \leq \varepsilon$. 
In particular the total sample complexity is given by: 
\begin{equation}
T \times M \times H \leq \frac{\tilde{C}_1^2}{\varepsilon^2} \times \frac{\tilde{C}_3^2 \tilde{C}_4}{\varepsilon^3} \times \frac{1}{1-\gamma} \log\left(\frac{\tilde{C}_2}{\varepsilon}\right) = \frac{\tilde{C}_1^2 \tilde{C}_3^2 \tilde{C}_4}{1-\gamma} \frac{1}{\varepsilon^5} \log\left(\frac{\tilde{C}_2}{\varepsilon}\right)\,,
\end{equation}
where~$\tilde{C}_1, \tilde{C}_2, \tilde{C}_3, \tilde{C}_4$ are defined in \eqref{eq:def-constants-tilde-C}. This concludes the proof. 

\subsubsection{Proof of Proposition~\ref{prop:fos-stoch-opt-guarantee-on-pol}}
\label{sec:proof-prop:fos-stoch-opt-guarantee-on-pol}

The proof follows standard optimization arguments, we provide a full proof for completeness. 
Throughout this proof, we use the shorthand notation $\bar{g}^t  := \E_t[g^t]\,.$

By $\beta$-smoothness of~$\Phi$, we have for every time step $t \geq 0$: 
\begin{align}
\label{eq:smoothnes-ineq}
\Phi(\pi^{t+1}) &\leq \Phi(\pi^{t}) + \langle \nabla \Phi(\pi^t), \pi^{t+1} - \pi^{t} \rangle + \frac{\beta}{2} \|\pi^{t+1} - \pi^{t}\|_2^2\,, \nonumber\\
&= \Phi(\pi^{t}) +  \langle g^t, \pi^{t+1} - \pi^{t} \rangle + \langle \nabla \Phi(\pi^t) - g^t, \pi^{t+1} - \pi^{t} \rangle + \frac{\beta}{2} \|\pi^{t+1} - \pi^{t}\|_2^2\,, \nonumber\\
&\leq \Phi(\pi^{t}) + \langle g^t, \pi^{t+1} - \pi^{t} \rangle +  \frac 12 \|\nabla \Phi(\pi^t) - g^t\|_2^2 + \frac{1+\beta}{2} \|\pi^{t+1} - \pi^{t}\|_2^2\,.
\end{align}
By characterization of the projection, the update rule of biased stochastic projected gradient descent yields:
\begin{equation}
\forall \pi \in \Pi, \langle \pi - \pi^{t+1} , (\pi^t - \eta g^t) - \pi^{t+1} \rangle \leq 0\,.
\end{equation}
Setting $\pi = \pi^t$ in the above inequality and rearranging, we obtain:
\begin{equation}
\label{eq:charac-proj-term}
\langle g^t, \pi^{t+1} - \pi^t \rangle \leq - \frac{1}{\eta} \|\pi^{t+1} - \pi^t\|_2^2 = - \eta \|\hat{G}^{\eta,\alpha}(\pi^t)\|_2^2\,.
\end{equation}
Moreover, using our bias bound assumption, we also have: 
\begin{equation}
\label{eq:interm-bias-var}
\frac 12 \|\nabla \Phi(\pi^t) - g^t\|_2^2 \leq \|\nabla \Phi(\pi^t) - \bar{g}^t\|_2^2 + \|g^t - \bar{g}^t\|_2^2 \leq \delta^2 + \|g^t - \bar{g}^t\|_2^2\,.
\end{equation}

Using the definition of the (stochastic) gradient mapping together with \eqref{eq:interm-bias-var} and \eqref{eq:charac-proj-term} in \eqref{eq:smoothnes-ineq}, we obtain: 
\begin{equation}
\Phi(\pi^{t+1}) \leq \Phi(\pi^{t}) - \eta \left(1 - \frac{\eta (1+ \beta)}{2}\right) \|\hat{G}^{\eta,\alpha}(\pi^t)\|_2^2 + \delta^2 + \|g^t - \bar{g}^t\|_2^2\,.
\end{equation}
Choosing $\eta \leq \frac{1}{1+\beta}$ gives $1 - \frac{\eta (1+ \beta)}{2} \geq \frac 12$ and we obtain by rearranging the above inequality: 
\begin{equation}
\|\hat{G}^{\eta,\alpha}(\pi^t)\|_2^2 \leq \frac{2}{\eta} (\Phi(\pi^{t}) - \Phi(\pi^{t+1})) + \delta^2 +   \|g^t - \bar{g}^t\|_2^2\,.
\end{equation}
Taking total expectation and using the variance bound assumption gives: 
\begin{equation}
\E[\|\hat{G}^{\eta,\alpha}(\pi^t)\|_2^2 ] \leq \frac{2}{\eta} \E[\Phi(\pi^{t}) - \Phi(\pi^{t+1})] + \delta^2 + \sigma^2\,.
\end{equation}
It remains to sum up the above inequality for $t \in [T-1]$ to obtain: 
\begin{equation}
\frac{1}{T} \sum_{t=0}^{T-1} \E[\|\hat{G}^{\eta,\alpha}(\pi^t)\|_2^2 ]  \leq \frac{2}{\eta T} \E[\Phi(\pi^{0}) - \Phi(\pi^{T})] + \delta^2 + \sigma^2 \leq \frac{2 (\Phi_{\max} - \Phi_{\min})}{\eta T} + \delta^2 + \sigma^2\,,
\end{equation}
and this concludes the proof. 

\subsubsection{Proof of Proposition~\ref{prop:relate-fos-to-grad-mapping-stoch}}
\label{sec:prop-relate-fos-to-grad-mapping-stoch-proof}

In this proof, we use the shorthand notation $g := g^t, \pi = \pi^t, \pi^{+} = \pi^{t+1}\,.$ 

We use the following decomposition of our term of interest for any policy $\pi' \in \Pi_{\alpha}$, 
\begin{equation}
\label{eq:decomp-proof-fos-gradmap}
    \langle \pi' - \pi^{+}, \nabla \Phi(\pi^{+}) \rangle 
    = \langle \pi' - \pi^{+}, g\rangle 
    + \langle \pi' - \pi^{+}, g - \nabla \Phi(\pi) \rangle
    + \langle \pi' - \pi^{+},  \nabla \Phi(\pi) - \nabla \Phi(\pi^{+}) \rangle\,.
\end{equation}

We now control each one of the terms in the above decomposition. 

For the first term we use the characterization of the projection that for any policy $\pi' \in \Pi^{\alpha}\,,$  
\begin{equation}
\langle \pi' - \pi^{+}, (\pi + \eta g) - \pi^+ \rangle \leq 0\,, 
\end{equation}
which yields the following after rearranging the inequality and using the Cauchy-Schwarz inequality: 
\begin{equation}
\label{eq:term1-proof-fos-gradmap}
\langle \pi' - \pi^{+}, g \rangle \leq - \frac{1}{\eta} \langle \pi' - \pi^{+}, \pi - \pi^+ \rangle =  \langle \pi' - \pi^{+}, \hat{G}^{\eta, \alpha}(\pi) \rangle \leq \|\pi - \pi'\|_2 \cdot \|\hat{G}^{\eta, \alpha}(\pi)\|_2 
\leq 2 \sqrt{|\cS|} \cdot \|\hat{G}^{\eta, \alpha}(\pi)\|_2\,.
\end{equation}

For the second term involving the bias of the gradient estimator, we immediately have 
\begin{equation}
\label{eq:term2-proof-fos-gradmap}
\langle \pi' - \pi^{+}, g - \nabla \Phi(\pi) \rangle \leq \| \pi' - \pi^{+}\|_2 \cdot \|g - \nabla \Phi(\pi)\|_{2} \leq \| \pi' - \pi^{+}\|_1 \cdot \|g - \nabla \Phi(\pi)\|_{2} \leq 2 \sqrt{|\cS|} \cdot \|g - \nabla \Phi(\pi)\|_{2}\,.
\end{equation}
For the last term, we use smoothness of the function~$\Phi$ as follows: 
\begin{equation}
\label{eq:term3-proof-fos-gradmap}
\langle \pi' - \pi^{+},  \nabla \Phi(\pi) - \nabla \Phi(\pi^{+}) \rangle \leq \|\pi' - \pi^{+}\|_2 \cdot \|\nabla \Phi(\pi) - \nabla \Phi(\pi^{+})\|_2 \leq 2 \sqrt{|\cS|} \beta \|\pi - \pi^{+}\|_2 = 2 \eta \beta  \sqrt{|\cS|} \|\hat{G}^{\eta, \alpha}(\pi)\|_2\,.
\end{equation}
Combining \eqref{eq:term1-proof-fos-gradmap}, \eqref{eq:term2-proof-fos-gradmap} and \eqref{eq:term3-proof-fos-gradmap} in \eqref{eq:decomp-proof-fos-gradmap}, we obtain for any $\pi' \in \Pi^{\alpha},$
\begin{equation}
\langle \pi' - \pi^{+}, \nabla \Phi(\pi^{+}) \rangle  \leq 2 (1+ \eta \beta)  \sqrt{|\cS|} \cdot \|\hat{G}^{\eta, \alpha}(\pi)\|_2 + 2 \sqrt{|\cS|} \cdot \|g - \nabla \Phi(\pi)\|_{2}\,.
\end{equation}

Now for any $\pi_i' \in \Pi_i$, set $\pi' = (\pi_i', \pi^{+}_{-i})$ in the previous inequality. Define $\pi_i^{\alpha} := (1-\alpha) \pi_i' + \frac{\alpha}{|\cA_i|} \in \Pi_i^{\alpha}$. Hence, we get for any $\pi_i' \in \Pi_i,$
\begin{align}
\label{eq:interm-ineq-proof-fos-gradmap}
\langle \pi_i' - \pi_i^{+}, \nabla_{\pi_i} \Phi(\pi^{+}) \rangle 
&= \langle \pi_i^{\alpha} - \pi_i^{+}, \nabla_{\pi_i} \Phi(\pi^{+}) \rangle  + \langle \pi_i' - \pi_i^{\alpha}, \nabla_{\pi_i} \Phi(\pi^{+}) \rangle\\
&\leq 2 (1+ \eta \beta)  \sqrt{|\cS|} \cdot \|\hat{G}^{\eta, \alpha}(\pi)\|_2 + 2 \sqrt{|\cS|} \cdot \|g - \nabla \Phi(\pi)\|_{2} + \alpha \left\langle \pi_i' - \frac{1}{|\cA_i|}, \nabla_{\pi_i} \Phi(\pi^{+}) \right\rangle\,.
\end{align}

We bound the last term in the above inequality as follows using our policy gradient theorem (see Proposition~\ref{prop:individual-pg}):
\begin{align}
\label{eq:last-term-bound-fos-gradmap}
\left\langle \pi_i' - \frac{1}{|\cA_i|}, \nabla_{\pi_i} \Phi(\pi^{+}) \right\rangle  
&= \sum_{s \in \cS, a_i \in \cA_i} \left(\pi_i'(a_i|s) - \frac{1}{|\cA_i|}\right) \cdot [\nabla_{\pi_i} \Phi(\pi^+)]_{s,a_i} \nonumber\\
&= \frac{1}{1-\gamma} \sum_{s \in \cS, a_i \in \cA_i} \left(\pi_i'(a_i|s) - \frac{1}{|\cA_i|}\right) d_{\rho}^{\pi^{+}}(s) \sum_{j=1}^N \bar{Q}_{s,a_i}^{\pi^{+}}(\bar{r}_{i,j}^{\pi^{+}})\nonumber\\
&\leq \frac{\sum_{j=1}^N\|\bar{r}_{i,j}^{\pi^{+}}\|_{\infty}}{(1-\gamma)^2} \sum_{s \in \cS, a_i \in \cA_i} \left| \pi_i'(a_i|s) - \frac{1}{|\cA_i|} \right| \cdot d_{\rho}^{\pi^{+}}(s)\nonumber\\ 
&\leq \frac{2 N l_{\infty}}{(1-\gamma)^2}\,.
\end{align}

Using \eqref{eq:last-term-bound-fos-gradmap} in \eqref{eq:interm-ineq-proof-fos-gradmap} yields: 
\begin{equation}
\underset{i \in \mathcal{N}}{\max} \underset{\pi_i' \in \Pi_i}{\max} \langle \pi_i' - \pi_i^{+}, \nabla_{\pi_i} \Phi(\pi^{+}) \rangle 
\leq 2 (1+ \eta \beta)  \sqrt{|\cS|} \cdot \|\hat{G}^{\eta, \alpha}(\pi)\|_2 + 2 \sqrt{|\cS|} \cdot \|g - \nabla \Phi(\pi)\|_{2} + \frac{2 N l_{\infty}}{(1-\gamma)^2} \alpha\,.
\end{equation}
We conclude the proof by taking total expectation in the above inequality and using our assumptions in~\eqref{as:bound-bias-var} (and Jensen's inequality) to bound the term $\|g - \nabla \Phi(\pi)\|_{2}$ in expectation and obtain the desired inequality: 
\begin{equation}
\E\left[\underset{i \in \mathcal{N}}{\max}\underset{\pi_i' \in \Pi_i}{\max} \langle \pi_i' - \pi_i^{+}, \nabla_{\pi_i} \Phi(\pi^{+}) \rangle \right] \leq 2 \sqrt{|\cS|}(1+\eta \beta) \E[\|\hat{G}^{\eta, \alpha}(\pi)\|_2] + 2 \sqrt{|\cS|} (\delta + \sigma) + \frac{2 N l_{\infty}}{(1-\gamma)^2} \alpha\,.
\end{equation}

\section{AUXILIARY RESULTS: SINGLE-AGENT RL}

Let $\mathcal{M} := (\mathcal{S}, \mathcal{A}, P, r, \rho, \gamma)$ where $r \in \mathcal{R}^{|\mathcal{S}| \cdot |\mathcal{A}|}$ be a Markov Decision Process (MDP). Define the value and action-value functions as follows for every state-action pair $(s,a) \in \mathcal{S} \times \mathcal{A},$
\begin{align}
V_{\rho}^{\pi}(r) &:= \mathbb{E}_{\rho, \pi}\left[ \sum_{t=0}^{+\infty} \gamma^t r(s_t, a_t)  \right] \,,\quad
Q_{s,a}^{\pi}(r) := \mathbb{E}_{\rho, \pi}\left[ \sum_{t=0}^{+\infty} \gamma^t r(s_t, a_t) | s_0 = s, a_0 = a \right] = \langle  \lambda_{s,a}^{\pi},  r \rangle\,,
\end{align}
where $\lambda_{s,a}^{\pi}$ is the occupancy measure corresponding to the Dirac initial distribution putting all the weight on the state action pair~$(s,a)\,.$ 

\begin{proposition}[Policy Gradient Theorem]
\label{prop:classical-pg-thm}
Let $\mathcal{M}(r) := (\mathcal{S}, \mathcal{A}, P, r, \rho, \gamma)$ where $r \in \mathcal{R}^{|\mathcal{S}| \cdot |\mathcal{A}|}$ be an MDP parameterized by a reward function $r$. Then, we have for all $\pi \in \Pi$ and all $r \in \mathcal{R}^{|\mathcal{S}| \cdot |\mathcal{A}|}$, 
\begin{equation}
\nabla_{\pi} V_{\rho}^{\pi}(r) = \frac{1}{1-\gamma} \sum_{s \in \mathcal{S}} \sum_{a \in \mathcal{A}} \lambda_{\rho}^{\pi}(s,a)\, Q_{s,a}^{\pi}(r) \, \nabla \log \pi(a|s)\,.
\end{equation}
Moreover, for the direct policy parameterization, we have for any $s \in \mathcal{S}, a \in \mathcal{A},$ 
\begin{equation}
\label{eq:pg-expression-direct-pp}
\frac{\partial V_{\rho}^{\pi}(r)}{\partial \pi(a|s)} = \frac{1}{1-\gamma} d_{\rho}^{\pi}(s) Q_{s,a}^{\pi}(r)\,.
\end{equation}
\end{proposition}

\begin{proposition}
\label{prop:lipschitzness-value-wrt-rewards}
For any policy~$\pi \in \Pi,$ for any reward functions~$r, r' \in \mathbb{R}^{|\mathcal{S}| \times |\mathcal{A}|}$, we have 
\begin{equation}
\| \nabla_{\pi} V_{\rho}^{\pi}(r) - \nabla_{\pi} V_{\rho}^{\pi}(r')\|_2 \leq \frac{\sqrt{|\mathcal{A}|}}{(1-\gamma)^{2}} \|r - r'\|_{\infty}\,.
\end{equation}
\end{proposition}

\begin{proof}
Using \eqref{eq:pg-expression-direct-pp}, we have 
\begin{align}
\label{eq:interm-diff-r}
\left|\frac{\partial V_{\rho}^{\pi}(r)}{\partial \pi(a|s)} - \frac{\partial V_{\rho}^{\pi}(r')}{\partial \pi(a|s)}\right| 
&= \frac{1}{1-\gamma} d_{\rho}^{\pi}(s) |Q_{s,a}^{\pi}(r) - Q_{s,a}^{\pi}(r')| \nonumber\\
&= \frac{1}{(1-\gamma)^2} d_{\rho}^{\pi}(s) |\langle \lambda_{s,a}^\pi, r - r'\rangle| \nonumber\\
&\leq \frac{1}{(1-\gamma)^2}  d_{\rho}^{\pi}(s) \|\lambda_{s,a}^\pi\|_1 \cdot \|r - r'\|_{\infty} \nonumber\\
&= \frac{1}{(1-\gamma)^2}  d_{\rho}^{\pi}(s) \|r - r'\|_{\infty}\,,
\end{align}
where for all $s' \in \mathcal{S}, a' \in \mathcal{A}$, 
$\lambda_{s,a}^\pi(s', a') = (1-\gamma) \sum_{t=0}^{+\infty} \gamma^t \bP_{\delta_{s,a},\pi}(s_t = s', a_t = a')$ and $\delta_{s,a}$ is the Dirac measure at the state-action pair $(s,a)\,.$ 
It follows that 
\begin{align}
\|\nabla_{\pi} V_{\rho}^{\pi}(r) - \nabla_{\pi} V_{\rho}^{\pi}(r')\|_2^2 &= \sum_{s,a} \left|\frac{\partial V_{\rho}^{\pi}(r)}{\partial \pi(a|s)} - \frac{\partial V_{\rho}^{\pi}(r')}{\partial \pi(a|s)} \right|^2\\
&\leq \frac{1}{(1-\gamma)^4} \sum_{s,a}  d_{\rho}^{\pi}(s)^2 \cdot \|r - r'\|_{\infty}^2\\
&\leq \frac{|\mathcal{A}|}{(1-\gamma)^4} \|r- r'\|_{\infty}^2\,, 
\end{align}
where the first inequality uses \eqref{eq:interm-diff-r} and the last inequality stems from recalling that $0 \leq  d_{\rho}^{\pi}(s) \leq 1$ which leads to $\sum_{s,a}  d_{\rho}^{\pi}(s)^2 \leq \sum_{s,a}  d_{\rho}^{\pi}(s) = |\mathcal{A}|\,.$ 
\end{proof}

\begin{proposition}[Variational Gradient Domination, Lemma~4 in \cite{agarwal-et-al21jmlr}]
\label{prop:var-grad-dom-mdps}
Let $\mathcal{M} := (\mathcal{S}, \mathcal{A}, P, r, \rho, \gamma)$ be an MDP and let $\pi^{\star}$ be an optimal policy. Then for any policy~$\pi \in \Pi$ and any state distribution~$\mu \in \Delta(\mathcal{S})$ we have 
\begin{equation}
V_{\rho}^{\pi^{\star}}(r) - V_{\rho}^{\pi}(r) \leq \left\| \frac{d_{\rho}^{\pi^{\star}}}{d_{\mu}^{\pi}} \right\|_{\infty}\, \max_{\pi' \in \Pi} \langle \nabla_{\pi_i} V_{\rho}^{\pi}(r), \pi' - \pi \rangle\,,
\end{equation}
where we suppose that $d_{\mu}^{\pi}(s) > 0$ for all $s \in \mathcal{S}\,.$ 
\end{proposition}

\begin{proposition}
\label{prop:lip-state-occup-measures}
For any policies $\pi, \tilde{\pi} \in \Pi$, we have 
\begin{equation}
\|d_{\rho}^{\pi} -  d_{\rho}^{\tilde{\pi}}\|_1 \leq \frac{\gamma}{1-\gamma} \E_{s \sim d_{\rho}^{\pi}} \left[ \|\pi(\cdot|s) - \tilde{\pi}(\cdot|s)\|_1 \right] \leq \frac{\gamma}{1-\gamma} \|\pi - \tilde{\pi}\|_1\,. 
\end{equation}
\end{proposition}

\begin{proof}
By definition, the state occupancy measures $d_{\rho}^{\pi}$ and $d_{\rho}^{\pi'}$ satisfy the Bellman consistency equations, i.e. for every state $s \in \mathcal{S}$, 
\begin{align}
d_{\rho}^{\pi}(s) &= \rho(s) + \gamma \sum_{s' \in \cS, a' \in \cA} P(s|s',a') \pi(a'|s') d_{\rho}^{\pi}(s')\,,\\
d_{\rho}^{\tilde{\pi}}(s) &= \rho(s) + \gamma \sum_{s' \in \cS, a' \in \cA} P(s|s',a') \tilde{\pi}(a'|s') d_{\rho}^{\tilde{\pi}}(s')\,.
\end{align}
As a consequence, we have the following upper bound: 
\begin{align}
|d_{\rho}^{\pi}(s) - d_{\rho}^{\tilde{\pi}}(s)| &\leq \gamma \sum_{s' \in \cS, a' \in \cA} P(s|s',a') |\pi(a'|s') d_{\rho}^{\pi}(s') -  \tilde{\pi}(a'|s') d_{\rho}^{\tilde{\pi}}(s')|\,.
\end{align}
Therefore we obtain by summing the above inequality over states and actions: 
\begin{align}
\|d_{\rho}^{\pi} - d_{\rho}^{\tilde{\pi}}\|_1 = \sum_{s \in \cS} |d_{\rho}^{\pi}(s) - d_{\rho}^{\tilde{\pi}}(s)| 
&\leq \gamma \sum_{s' \in \cS, a' \in \cA} \sum_{s \in \cS} P(s|s',a') |\pi(a'|s') - \tilde{\pi}(a'|s')|  d_{\rho}^{\pi}(s')\nonumber\\
&= \gamma \sum_{s' \in \cS, a' \in \cA}  |\pi(a'|s') d_{\rho}^{\pi}(s') -  \tilde{\pi}(a'|s') d_{\rho}^{\tilde{\pi}}(s')| \nonumber\\
&\leq \gamma \sum_{s' \in \cS, a' \in \cA} d_{\rho}^{\pi}(s') \cdot |\pi(a'|s') - \tilde{\pi}(a'|s')| + \gamma \sum_{s' \in \cS, a' \in \cA} \tilde{\pi}(a'|s') \cdot |d_{\rho}^{\pi}(s') - d_{\rho}^{\tilde{\pi}}(s')|\nonumber\\
&=  \gamma \E_{s \sim d_{\rho}^{\pi}} \left[ \|\pi(\cdot|s) - \tilde{\pi}(\cdot|s)\|_1 \right] + \gamma \|d_{\rho}^{\pi} - d_{\rho}^{\tilde{\pi}}\|_1\,.
\end{align}
Rearranging the above inequality gives the desired result. 
\end{proof}

\newpage
\section{NUMERICAL SIMULATIONS}

While our contributions are mainly theoretical, we performed simple simulations to illustrate the performance of our PG algorithm (Algorithm~\ref{algo}) in the stochastic setting. We consider three different examples from section~\ref{subsec:examples} (a, d and e) in the potential setting with $N=3$ agents. We set $\gamma=0.95$, $\eta=0.01$, $H=20$, $A_i = 5$ (number of different actions) on a $5 \times 5$ grid environment with $512$ episodes per iteration. In all the different simulations, we observe that the potential is increasing and the occupancy/Nash gaps are reducing as expected along the iterations of the algorithm. 

\begin{table}[h]
\centering
\caption{Multi-agent Imitation Learning (example a in section~\ref{subsec:examples})}
\resizebox{0.6\columnwidth}{!}{
\label{tab:multiagent_imitation}
\begin{tabular}{cccc}
\toprule
\textbf{Iter} & \textbf{Potential} & \textbf{Occupancy Gap} & \textbf{KL Occupancy Gap} \\
\midrule
1   & -1.0712 & 1.2263 & 1.0712 \\
20  & -0.9606 & 1.1851 & 0.9606 \\
40  & -0.8469 & 1.1373 & 0.8469 \\
60  & -0.6965 & 1.0277 & 0.6965 \\
100 & -0.4607 & 0.8135 & 0.4607 \\
120 & -0.2687 & 0.6321 & 0.2687 \\
160 & -0.1932 & 0.5467 & 0.1932 \\
180 & -0.2173 & 0.5782 & 0.2173 \\
200 & -0.1591 & 0.4882 & 0.1591 \\
\bottomrule
\end{tabular}}
\end{table}

\begin{table}[h]
\centering
\caption{Team Coverage and Distribution Matching (example d in section~\ref{subsec:examples})}
\label{tab:team_coverage_distribution}
\begin{tabular}{ccccc}
\toprule
\textbf{Iter} & \textbf{Potential} & \textbf{Nash Gap} & \textbf{Occupancy Gap} & \textbf{KL Occupancy Gap} \\
\midrule
1   & -1.8612 & 95.0888 & 1.3509 & 1.8612 \\
10  & -1.7698 & 92.1758 & 1.3128 & 1.7698 \\
20  & -1.6501 & 88.4389 & 1.2765 & 1.6501 \\
50  & -1.2508 & 74.4521 & 1.1924 & 1.2508 \\
80  & -0.8294 & 55.9668 & 1.0197 & 0.8294 \\
100 & -0.4804 & 42.0560 & 0.7538 & 0.4804 \\
130 & -0.1305 & 22.7316 & 0.3542 & 0.1305 \\
160 & -0.0148 & 7.8765  & 0.1297 & 0.0148 \\
190 & -0.0040 & 5.3988  & 0.0629 & 0.0040 \\
\bottomrule
\end{tabular}
\end{table}

\begin{table}[h]
\centering
\caption{Collective Exploration (example e in section~\ref{subsec:examples})}
\label{tab:collective_exploration}
\begin{tabular}{cccc}
\toprule
\textbf{Iter} & \textbf{Potential} & \textbf{Nash Gap} & \textbf{Occupancy Gap} \\
\midrule
1   & 2.9909 & 95.0888 & 1.3434 \\
10  & 3.0684 & 92.1724 & 1.3079 \\
30  & 3.3050 & 84.3933 & 1.2522 \\
50  & 3.5616 & 74.7144 & 1.1992 \\
80  & 3.9804 & 56.6424 & 1.0325 \\
100 & 4.3174 & 43.2140 & 0.7834 \\
150 & 4.7797 & 15.6444 & 0.2363 \\
200 & 4.8165 & 7.7096  & 0.1010 \\
210 & 4.8198 & 7.6516  & 0.0872 \\
\bottomrule
\end{tabular}
\end{table}

\end{document}